\documentclass[journal, 10pt, onecolumn, draftclsnofoot]{IEEEtran}

\usepackage[reqno]{amsmath}
\usepackage{graphicx}
\usepackage{bbm}
\usepackage{mathrsfs}
\usepackage{stmaryrd}
\usepackage{graphics}
\usepackage{acronym}
\usepackage{longtable}
\usepackage{mathtools}
\usepackage{times}
\usepackage{setspace}
\usepackage{cite}
\usepackage{array}
\usepackage{subfigure}
\usepackage{amsmath,amsthm}
\usepackage{amssymb}
\usepackage{wasysym,url}
\usepackage{fixltx2e}
\usepackage{setspace,float}
\usepackage{color}
\usepackage{cases,bm}
\usepackage[inline]{enumitem}
\usepackage{tikz}
\usepackage{hyperref}
\usepackage{mathtools,cuted}
\usepackage{multirow}
\usepackage{array}
\usepackage{booktabs}
\usepackage[noabbrev]{cleveref}
\crefformat{figures}{(#2#1#3)}
\usepackage{xcolor}
\usepackage{tabularx}
\usepackage{booktabs}
\usepackage[linesnumbered,ruled,vlined]{algorithm2e}
\DontPrintSemicolon

\makeatletter
\newcommand{\leqnomode}{\tagsleft@true\let\veqno\@@leqno}
\newcommand{\reqnomode}{\tagsleft@false\let\veqno\@@eqno}
\makeatother

\newcommand{\abs}[1]{\left\vert{#1}\right\vert}
\newcommand{\bb}[1]{\mathbb{#1}}

\newcommand{\bd}[1]{\mathbf{#1}}
\newcommand{\bld}[1]{\boldsymbol{#1}}

\newcommand{\cl}[1]{\mathcal{#1}}

\newcommand{\sT}{\mathscr{T}}

\newcommand{\TT}{\top}
\newcommand{\jj}{\mathrm{j}}
\newcommand{\dd}{\mathrm{d}}

\newcommand{\orcid}[1]{\href{https://orcid.org/#1}{\textcolor[HTML]{A6CE39}{\aiOrcid}}}

\newtheorem{lemma}{Lemma}
\newtheorem{proposition}{Proposition}

\newtheorem{definition}{Definition}

\definecolor{lime}{HTML}{A6CE39}
\DeclareRobustCommand{\orcidicon}{%
	\begin{tikzpicture}
	\draw[lime, fill=lime] (0,0)
	circle [radius=0.16]
	node[white] {{\fontfamily{qag}\selectfont \tiny ID}};
	\draw[white, fill=white] (-0.0625,0.095)
	circle [radius=0.007];
	\end{tikzpicture}
	\hspace{-2mm}
}

\foreach \x in {A, ..., Z}{%
	\expandafter\xdef\csname orcid\x\endcsname{\noexpand\href{https://orcid.org/\csname orcidauthor\x\endcsname}{\noexpand\orcidicon}}
}

\title{Neuromorphic Sampling of Finite-Rate-of-Innovation Signals}
\title{Neuromorphic Sampling of Sparse Signals}
\author{\IEEEauthorblockN{
Abijith Jagannath Kamath,\orcidA{}~\IEEEmembership{Student Member,~IEEE}, and
Chandra Sekhar Seelamantula,\orcidC{}~\IEEEmembership{Senior Member,~IEEE}}

\thanks{
A.~J.~Kamath and C.~S.~Seelamantula are with the Department of Electrical Engineering, Indian Institute of Science, Bangalore (Email: \{abijithj, css\}@iisc.ac.in).\\
\indent The figures in this paper are in colour in the electronic version.
}}

\begin{document}
\maketitle

%


\begin{abstract}
Neuromorphic sampling is a bioinspired and opportunistic analog-to-digital conversion technique, where the measurements are recorded only when there is a significant change in the signal amplitude. Neuromorphic sampling has paved the way for a new class of vision sensors called \emph{event cameras} or \emph{dynamic vision sensors (DVS)}, which consume low power, accommodate a high-dynamic range, and provide sparse measurements with high temporal resolution making it convenient for downstream inference tasks. In this paper, we consider neuromorphic sensing of signals with a finite rate of innovation (FRI) --- including a stream of Dirac impulses, sum of weighted and time-shifted pulses, and piecewise-polynomial functions. We consider a sampling-theoretic approach and leverage the close connection between neuromorphic sensing and time-based sampling, where the measurements are encoded temporally. Using Fourier-domain analysis, we show that perfect signal reconstruction is possible via parameter estimation using high-resolution spectral estimation methods. We develop a kernel-based sampling approach, which allows for perfect reconstruction with a sample complexity equal to the rate of innovation of the signal. We provide sufficient conditions on the parameters of the neuromorphic encoder for perfect reconstruction. Furthermore, we extend the analysis to multichannel neuromorphic sampling of FRI signals, in the single-input multi-output (SIMO) and multi-input multi-output (MIMO) configurations. We show that the signal parameters can be jointly estimated using multichannel measurements. Experimental results are provided to substantiate the theoretical claims.
\end{abstract}





\section{Introduction}
\IEEEPARstart{N}{euromorphic} sampling is a bioinspired continuous-time signal acquisition technique and is different from uniform sampling. In traditional analog-to-digital converters, continuous-time signals are measured at periodic intervals called the sampling interval, which depends only on the class of input signals \cite{shannon1949communication,unser2000sampling}. In neuromorphic sampling, measurements are recorded on the temporal axis only when there is a significant change in the signal, denoting an event, thus, giving rise to a signal-dependent sampling strategy. Thence, neuromorphic sampling is also called \emph{event-driven} sampling. In neuromorphic sampling, measurements, also known as \emph{events}, are 2-tuples comprising the time-instant of the change in the signal, and a one-bit polarity (ON/OFF), which indicates the sign of the change. Sensitivity to only differences of the signal makes the acquisition scheme \emph{opportunistic} \cite{guan2007opportunistic}, {\it i.e.}, in intervals of inactivity in the signal, the acquisition scheme does not record any events.\\
\indent Neuromorphic sampling has opened up avenues for a novel class of audio \cite{chan2007aer} and vision sensors \cite{mahowald1994silicon,delbruck2010sensors,lichtsteiner2008sensor,brandii2014sensor}. Vision sensors or \emph{event cameras}, such as DAVIS346 \cite{gallego2022event}, are 2D spatial arrays of encoders that function on the neuromorphic principle. Event cameras are energy efficient ($\approx 1$ mW), since the sensor array asynchronously measures events without relying on a global clock. They are designed to provide compressed measurements at source and minimize redundancy in the representation. They also capture high-dynamic-range (HDR) signals ($\approx 120$ dB) as the acquisition is sensitive to only finite-differences in the signal. Event cameras have a high temporal resolution ($\approx 1~\mu$s), which makes them perfectly suited for sensing ultrafast changes in the visual scene.\\
\indent The development of neuromorphic sensors is inspired by the functioning of the human vision system \cite{mahowald1994silicon,van2003selective}. An image that is formed on the retina is transmitted to the lateral geniculate nucleus (LGN) through \emph{nasal} and \emph{temporal} connections, and thereafter to the primary visual cortex. The visual information is encoded in the LGN and carried via two types of visual pathways --- parvocellular and magnocellular, where the former is responsible for encoding sustained visual stimulus such as the subject background and peripheral vision, and the latter is responsible for encoding transient or dynamic vision, {\it i.e.}, visual events that are fast-changing. Neuromorphic vision sensors are inspired by the magnocellular sensing mechanism, whereas standard frame-based cameras mimic the parvocellular mechanism.



	

\subsection{Sparse Sampling Meets Sparse Signals}
In this paper, we address the problem of sampling and reconstruction of continuous-time signals of the type:
\begin{equation}\label{eq:signal_model}
	x(t) = \sum_{k=0}^{K-1}a_k\varphi(t-\tau_k), \; 0\leq t< T,
\end{equation}
where $0<\tau_0<\tau_1<\cdots<\tau_{K-1}<T$ are ordered shift parameters, $T\in\bb R$ is the duration of the signal and $\varphi\in L_2(\bb R)$ is a known pulse. Figure~\ref{fig:sampling_illustration} illustrates neuromorphic sampling of $K=17$ pulses. The signal $x(t)$ is completely characterized by the $K$ \emph{coefficients} $\{a_k\in\bb R\}_{k=0}^{K-1}$ and the $K$ \emph{support parameters} $\{\tau_k\in[0,T[\}_{k=0}^{K-1}$. The signal $x$ is said to have a finite rate of innovation (FRI) of $\frac{2K}{T}$ \cite{vetterli2002sampling}. FRI signals are analog counterparts of sparse vectors \cite{do2008union}, and are perfectly represented using $2K$ measurements \cite{blu2008cadzow}, {\it i.e.}, measurements recorded at the rate of innovation of the signal.\\
\indent In this paper, we connect sparse signals to sparse sampling, by using neuromorphic encoders to sample FRI signals. Within this opportunistic sampling framework, measurements are made only when there is a significant activity in the signal. Figure~\ref{fig:sampling_illustration} demonstrates neuromorphic sampling of an FRI signal $x(t)$ using a lowpass filter. The events are \emph{dense} in the intervals where there is a significant change in the signal, and are \emph{sparse} in intervals of inactivity in the signal. On the other hand, if one were to use uniform sampling, the samples are recorded independently of the signal at a fixed rate.\\
\indent Several works have utilized the low sampling requirement of FRI signals for sub-Nyquist signal processing. The key idea is similar to the echolocation principle \cite{lee1992common}. In active sensing, a known pulse $\varphi(t)$ is transmitted, and the received signal is modelled as $x(t)$ (cf.~Eq.~\eqref{eq:signal_model}), which is a linear combination of time-shifted versions of $\varphi(t)$, where the parameters $\{(a_k,\tau_k)\}$ are related to the object that is being sensed. For instance, in the case of pulse-Doppler radar imaging \cite{skolnik1980introduction,richards2014fundamentals,bajwa2011identification,ilan2014radar,rudresh2017radar}, the delays correspond to the distance of each object and the phase of the complex amplitudes are related to the Doppler of each target. In passive sensing modalities, for instance, in radioastronomy \cite{pan2016towards}, the received signal is modelled as an FRI signal as given in Eq.~\eqref{eq:signal_model}. Other applications include ultrasound imaging \cite{tur2011innovation}, Fourier-domain optical coherence tomography \cite{mulleti2014fdoct}, underwater sonar imaging \cite{srinath2020nyquist}, time-of-flight imaging \cite{bhandari2016signal}, and ground penetrating radar \cite{rudresh2020sampling}. Viewing the FRI signal model in Eq.~\eqref{eq:signal_model} in the Fourier domain also finds applications in signal interpolation tasks \cite{pan2014sampling,ongie2015super,mulleti2016effrip,kamath2019curves}.


\subsection{Related Literature}
Neuromorphic sampling is closely related to time-encoding of continuous-time signals \cite{lazar2004sensitivity,gontier2014sampling,adam2022timing}. Time-encoding of FRI signals was first addressed by Alexandru and Dragotti \cite{alexandru2019crossing,alexandru2019reconstructing}, who considered integrate-and-fire time-encoding of a stream of Dirac impulses using an exponential-reproducing kernel. They developed a sequential reconstruction algorithm that recovers one impulse at a time. Their reconstruction strategy requires the support of the sampling kernel to be smaller than the spacing between two consecutive impulses, {\it i.e.}, the sampling kernel becomes signal-dependent, which is a limitation of the method. Satisfying the support constraint also requires sampling beyond the rate of innovation. Hilton {\it et al.} \cite{hilton2021guaranteed} use the sequential reconstruction algorithm on a stream of Dirac impulses filtered using an alpha synaptic function, such that the sampling requirement for perfect reconstruction can be guaranteed by tuning the parameters of the time-encoding machine. We overcame the shortcomings of \cite{alexandru2019reconstructing} by proposing a Fourier-domain reconstruction approach \cite{rudresh2020time}, which allowed for reconstruction of FRI signals observed using arbitrary sampling kernels, with number of measurements of the order of the rate of innovation of the signal. Recently, Naaman {\it et al.} \cite{naaman2022fritem,naaman2021temhardware} addressed reconstruction of FRI signals from integrate-and-fire time-encoding in the presence of noise. Their method differs from that of \cite{rudresh2020time} in the sense that their reconstruction strategy removes terms in the forward linear transformation that cause instability. Kamath {\it et al.} also considered the problem of reconstruction of FRI signals from time-encoded measurements in the presence of noise \cite{kamath2021time}, using the generalized FRI approach \cite{pan2016towards}. Kamath and Seelamantula developed multichannel time-encoding of FRI signals, in the single-input, multi-output (SIMO) and multi-input, multi-output (MIMO) settings \cite{kamath2023multichannel}.\\
\indent In a companion paper, we developed a framework for neuromorphic sampling of continuous-time signals in shift-invariant spaces, and also provided a time-encoding model for neuromorphic sampling \cite{kamath2023neuromorphic}. Using the proposed time-encoding model, we determined the $t$-transform \cite{lazar2003recovery} (cf. Lemma~\ref{lem:ttransform}) that maps events to nonuniform amplitude samples of the signal. Thereafter, signal reconstruction is achieved using an alternating projections technique that enforces the signal prior and measurement consistency. In this paper, we consider neuromorphic sampling of a different signal class --- signals with a finite rate of innovation, thereby connecting sparse signals to sparse sampling.
\begin{figure}[!t]
	\centering
	\includegraphics[width=3.5in]{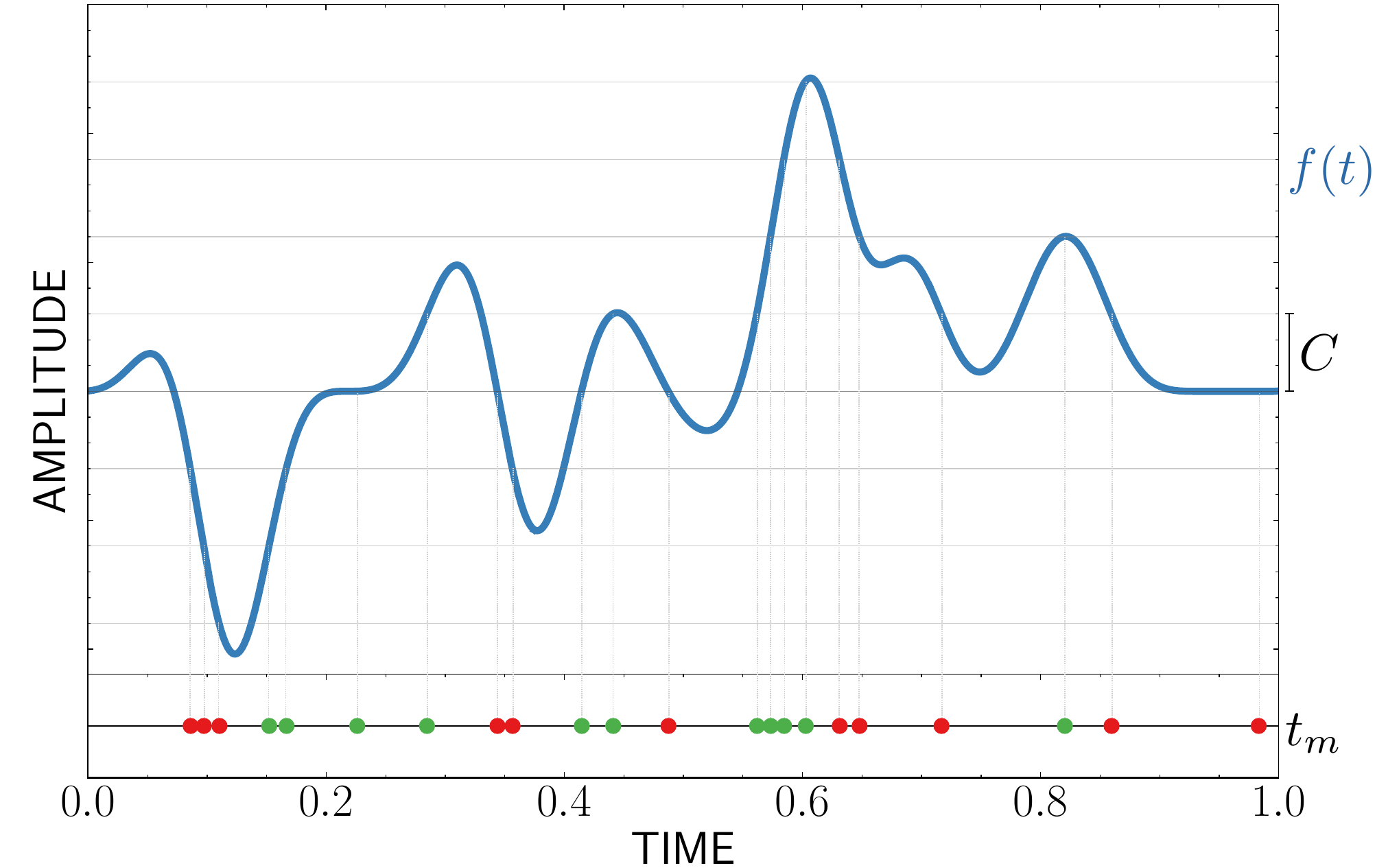}
	\caption{Neuromorphic sampling of FRI signals: a sum of $K=17$ weighted and time-shifted pulses as in Eq.~\eqref{eq:signal_model} is sampled using a neuromorphic encoder using the sum-of-modulated-splines sampling kernel. The event time-instants and the ON (green)/OFF (red) polarities, are sparse, while being sufficient for perfect reconstruction. Further, they occur only in intervals in the neighbourhood of the shift parameters, indicating energy-efficient encoding.}
	\label{fig:sampling_illustration}
\end{figure}


\subsection{Contributions of this Paper}
Neuromorphic sampling is signal-dependent, which is a perfect fit for sparse continuous-time signals that have several intervals of inactivity. FRI signals, such as a stream of Dirac impulses, piecewise polynomials, and signals composed of weighted and time-shifted pulses, naturally belong to this class. We develop a kernel-based acquisition scheme that allows for signal reconstruction using parameter estimation, and show that perfect reconstruction is possible with the number of measurements proportional to the rate of innovation of the signal. We also provide sufficient conditions on the parameters of the neuromorphic encoder to achieve perfect reconstruction (Section~\ref{sec:siso_sampling}). Furthermore, we extend the scheme to SIMO and MIMO configurations, and show that perfect reconstruction is possible using joint estimation techniques, and provide sufficient conditions (Section~\ref{sec:multichannel_sampling}).
Experimental results validate the theoretical claims.
Before proceeding further, we present the mathematical preliminaries (Section~\ref{sec:preliminaries}) required for the developments reported in this paper.



\section{Mathematical Preliminaries}
\label{sec:preliminaries}

\subsection{Neuromorphic Sampling}
In neuromorphic sampling, the signal is encoded as a stream of \emph{events}. Each event contains the time instant ({\it event time-instant}) at which the signal amplitude changes in magnitude by a threshold, along with the polarity of the change ({\it event polarity}). More concretely, a neuromorphic encoder is a mapping from continuous-time signals to a sequence of 2-tuples of the type $(t_m, p_m)\in\bb R\times \{-1,+1\}$, which denotes an event, where $\{t_m\}$ comprises the \emph{event time-instants} at which the signal changes in magnitude by a constant $C>0$ called the \emph{temporal contrast threshold}, and $\{p_m\}$ comprises the \emph{event polarities}: $+1$ for a positive change, and $-1$ for a negative change. The following definition encapsulates the functioning of the neuromorphic encoder \cite{kamath2023neuromorphic}.
\begin{definition}[Neuromorphic encoder] \label{def:neuromorphic_encoder}
A neuromorphic encoder $\sT_C$ is a device that maps a continuous-time signal $f(t)$ to a sequence of 2-tuples $\{(t_m, p_m)\}_{m\in\bb N}$ comprising the event instants $\{t_m\}$ and polarities of change $\{p_m\}$ such that
	\begin{equation}
	\begin{split}
		t_{m} &= \min \left\{t : |f(t) - f(t_{m-1})| = C, \; t>t_{m-1}\right\}, \\
		p_m &= \text{sgn}(f(t_m) - f(t_{m-1})),
	\end{split}
	\end{equation}
	where $C>0$ denotes the temporal contrast threshold of the neuromorphic encoder and $\text{sgn}(\cdot)$ denotes the signum function.
\end{definition}
Figure~\ref{fig:schematic_fri_sampler} shows the schematic of a neuromorphic encoder implemented using a comparator. The comparator accepts the difference between the input continuous-time signal $f(t)$ and its amplitude evaluated at the previous event time-instant $t_{m-1}$ against thresholds $\pm C$. If the difference exceeds $+C$ at time instant $t=t_m$, a positive polarity spike ($p_m=+1$) is produced at the output and denotes an `ON' event. Likewise, if the difference subceeds $-C$ at time instant $t=t_m$, a negative polarity spike ($p_m=-1$) is produced at the output and denotes an `OFF' event. Thus, the output of the comparator is a stream of Dirac impulses:
\[
	p(t) = \sum_{m\in\bb N} p_m\delta(t-t_m).
\]
The encoder output $\{(t_m,p_m)\}_{m\in\bb N}$ can be obtained from $p(t)$ using sub-Nyquist methods in \cite{seelamantula2010sub,vetterli2002sampling} or using a counter as deployed in \cite{brandii2014sensor,naaman2021temhardware}.\\
\indent The temporal measurements are, in general, nonuniformly spaced. One can define the nonuniform counterpart of the sampling interval, called the \emph{sampling density}, as follows:
\begin{equation}\label{eq:sampling_density}
	\mathfrak{D}(\{t_m\}_{m\in\bb N}) \overset{\text{def.}}{=} \sup_{m\in\bb N} |t_{m}-t_{m-1}|.
\end{equation}
A bounded sampling density is crucial for achieving perfect reconstruction of signals in shift-invariant spaces using filtering techniques as shown in \cite{kamath2023neuromorphic}. However, in the context of FRI signals, the perfect reconstruction condition relies on the minimum number of measurements matching the rate of innovation of the signal, and not on the sampling density. The number of measurements obtained over an interval depends on the dynamic range of the signal relative to the temporal contrast threshold $C$.
\begin{lemma}[Minimum number of measurements] \label{lem:min_measurements}
Consider a continuous, continuous-time signal $f(t)$, input to a neuromorphic encoder with temporal contrast threshold $C>0$. Let $[0, T]$ denote the measurement interval, and $\displaystyle f_{\max} = \max_{t\in [0,T]} f(t)$ and $\displaystyle f_{\min} = \min_{t\in [0,T]} f(t)$. To obtain at least $L$ events over $[0,T]$, the temporal contrast threshold must satisfy
\begin{equation}\label{eq:min_measurements}
	0 < C < \frac{f_{\max}-f_{\min}}{L}.
\end{equation}
\begin{proof}
The dynamic range of $f$ is $[f_{\min},f_{\max}]$. The values $f_{\min}$ and $f_{\max}$ exist and are well-defined, since $f$ is continuous and the measurement interval $[0, T]$ is compact. A measurement is recorded every time the signal changes by $C$. The maximum number of partitions of the interval $[f_{\min}, f_{\max}]$, each of width $C$, is $\displaystyle\left\lfloor \frac{f_{\max}-f_{\min}}{C} \right\rfloor$. By virtue of continuity of $f$, and invoking the intermediate-value theorem \cite{rudin1953principles}, we observe that the neuromorphic encoder generates at least one event in each partition. Therefore, at least $L$ events can be obtained by setting the temporal contrast threshold according to Eq.~\eqref{eq:min_measurements}. 
\end{proof}
\end{lemma}
\begin{figure}[!t]
	\centering
	\includegraphics[width=3.5in]{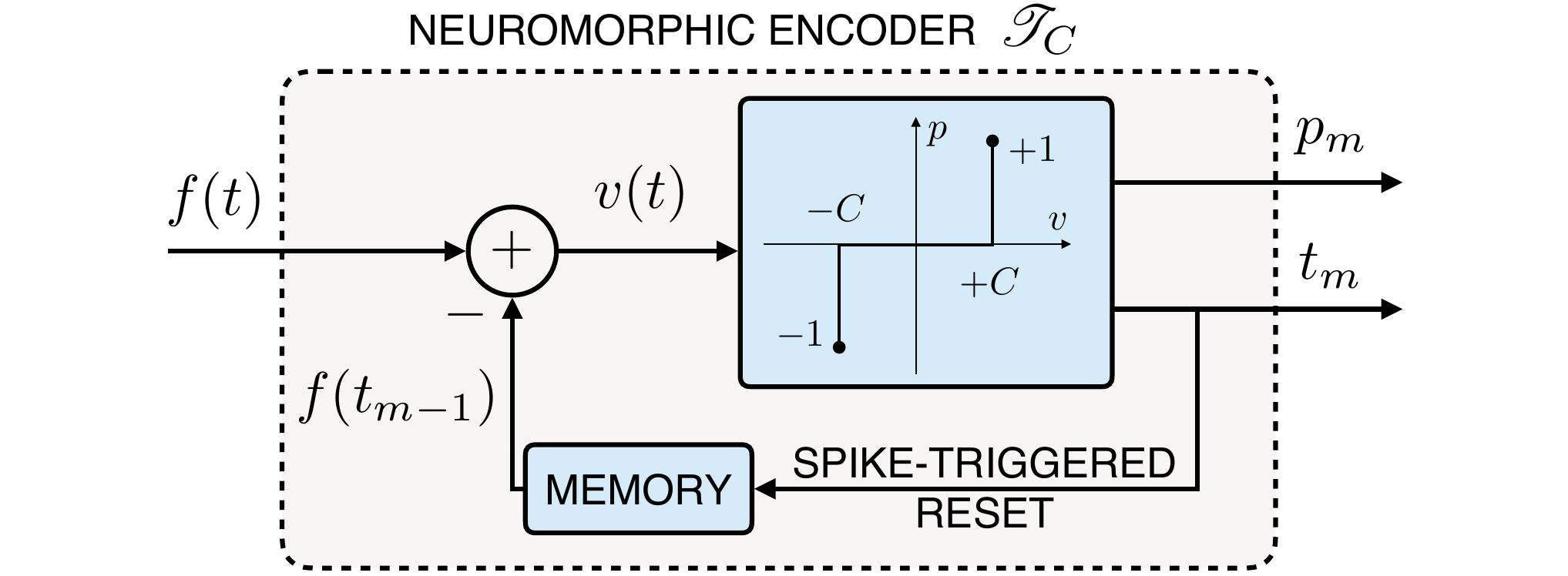}
	\caption{Schematic of a neuromorphic encoder with temporal contrast threshold $C$. The input $f(t)$ is offset by $f(t_{m-1})$ and compared against $\pm C$. The time instant $t_m$ and polarity $p_m$ (together called an event) constitute the output.}
	\label{fig:schematic_fri_sampler}
\end{figure}
\begin{figure*}[!t]
	\centering
	\includegraphics[width=6.75in]{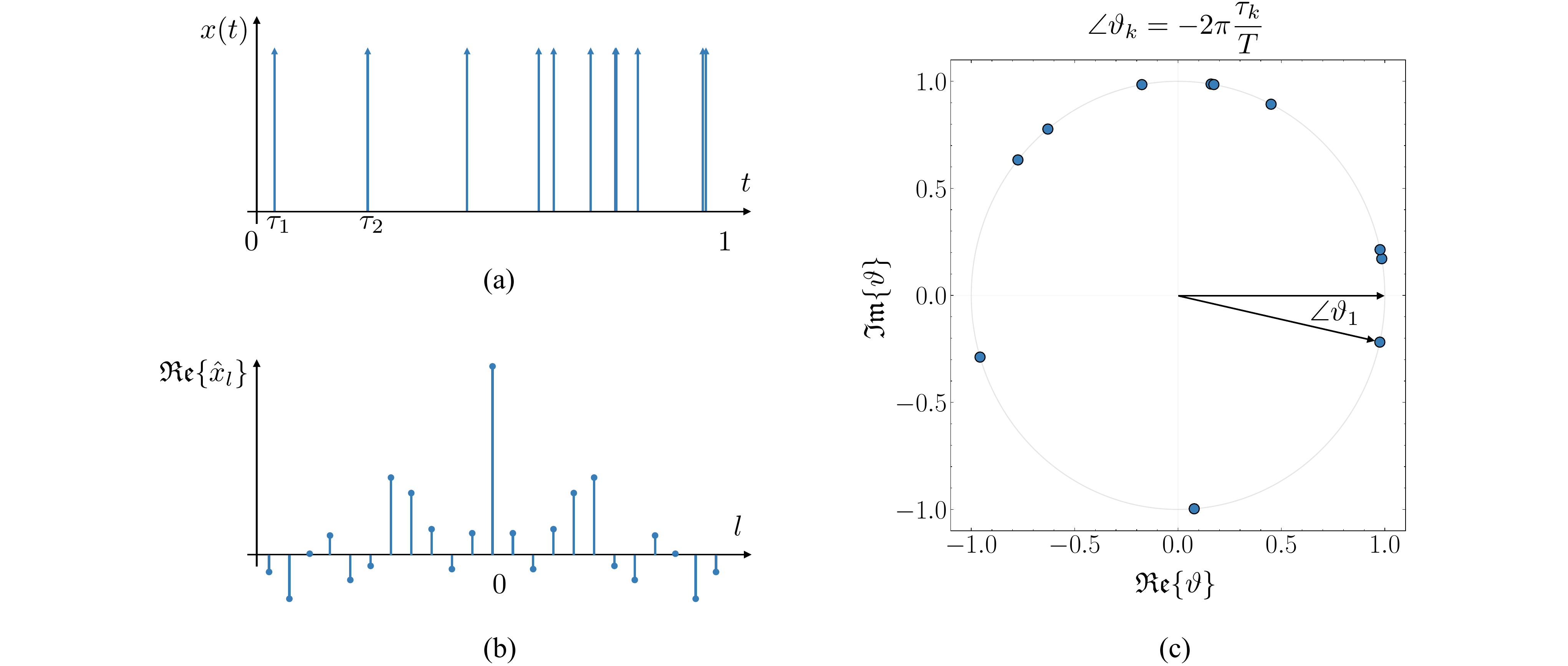}
	\caption[]{Illustration of Prony's method for recovery of support parameters: (a) shows a stream $K=11$ of Dirac impulses $x(t)$ as defined in Eq.~\eqref{eq:sparse_model} with equal amplitudes; (b) shows the real part of $\{\hat{x}_l\}$; and (c) shows the location of the zeros of the annihilating filter $\bld h$ on the complex plane. The Fourier coefficients have a sinusoidal structure whose frequencies are determined by the support parameters of $x(t)$. The locations of the zeros of $\bld h$ are in one-to-one correspondence with the support parameters of $x(t)$.}
	\label{fig:annihilating_filter_illustration}
\end{figure*}
Although neuromorphic sampling does not explicitly record amplitude measurements, it is possible to obtain amplitude samples of the signal from the output of the neuromorphic encoder, akin to the $t$-transform in time-based sampling \cite{lazar2003recovery,gontier2014sampling,kamath2022differentiate}. We recall the neuromorphic $t$-transform next.
\begin{lemma}[Neuromorphic $t$-transform \cite{kamath2023neuromorphic}] \label{lem:ttransform}
Consider a continuous, continuous-time signal $f(t)$ and let $\sT_C\{f\} = \{(t_m, p_m)\}_{m\in\bb N}$ denote the sequence of events generated by the neuromorphic encoder with temporal contrast threshold $C>0$. The samples of the signal at event time-instants $\{t_m\}_{m\in\bb N}$ are given by
\begin{equation}\label{eq:lemma2}
	f(t_m) = f(t_0) + C\sum_{i=1}^m p_i, \;\forall m\in\bb N,
\end{equation}
where $t_0<t_1$ is the reference time-instant with respect to which the events are measured, and $f(t_0)$ is the corresponding initial value.
\end{lemma}
Effectively, given an initial value, the $t$-transform provides a method to compute the amplitude samples from the events.


\subsection{Support Recovery Using Prony's Method}
\label{subsec:prony}
Consider the prototypical $K$-sparse, compactly-supported FRI signal composed of Dirac impulses:
\begin{equation}\label{eq:sparse_model}
	x(t) = \sum_{k=0}^{K-1} a_k\delta(t-\tau_k), \; 0\leq t\leq T,
\end{equation}
where $\{a_k\}_{k=0}^{K-1}$ are complex-valued coefficients and $0<\tau_0<\tau_1<\cdots<\tau_{K-1}<T$ are the support parameters, which completely characterize the signal. Recovery of the support parameters is a nonlinear problem, and is better understood in the Fourier domain. The support recovery problem is effectively the Fourier dual of the parametric spectral estimation problem \cite{stoica2000spectral}, because, streams of nonuniform Dirac impulses and a sum of complex exponentials are Fourier pairs. The Fourier duality is illustrated in Figure~\ref{fig:annihilating_filter_illustration}. The support parameters can be recovered using high-resolution spectral estimation (HRSE) techniques \cite{stoica2000spectral} such as Prony's annihilating filter method \cite{prony}, multiple signal classification (MUSIC) \cite{schmidt1986multiple}, and estimation of signal parameters via rotational invariant technique (ESPRIT) \cite{roy1989esprit}. The simplest of these techniques is Prony's method, which we recall next.\\
\indent Consider a periodized version of $x(t)$ with time-period $T$. Invoking the Fourier-series representation, we can write
\begin{equation}\label{eq:swce_dirac}
	x(t) = \sum_{l\in\bb Z} \underbrace{\frac{1}{T}\sum_{k=0}^{K-1} a_{k} e^{-\jj \omega_0 l \tau_{k}}}_{\hat{x}_l} e^{\jj \omega_0 l t}, \; t\in[0,T],
\end{equation}
where $\displaystyle \omega_0 = \frac{2\pi}{T}$, and $\hat{\boldsymbol{x}} = \{\hat{x}_l\}_{l\in\bb Z}$ are the Fourier coefficients.
The Fourier coefficients are in the form of a sum of weighted complex exponentials. Prony's annihilating filter method \cite{prony} requires at least $2K+1$ contiguous Fourier coefficients and employs a $(K+1)$-tap annihilating filter $\boldsymbol{h} = \{h_k\}_{k=0}^K\in\bb C^{K+1}$ with zeros $\{\vartheta_k = e^{-\jj 2\pi \tau_k/T}\}_{k=0}^{K-1}$, which are in one-to-one correspondence with the support parameters $\{\tau_{k}\}_{k=0}^{K-1}$. The $z$-transform of the annihilating filter is given by
\begin{equation*}
	H(z) = \sum_{k=0}^{K} h_{k} z^{-k} = h_{0}\prod_{k=0}^{K-1}(1-\vartheta_{k} z^{-1}).
\end{equation*}
The {\it annihilation property} can be verified as follows:
\begin{equation}
\begin{split}
	(\boldsymbol{h}*\hat{\boldsymbol{x}})_l &= \sum_{k\in\bb Z} h_{k} \hat{x}_{l-k} = \sum_{k=0}^{K} h_{k}  \frac{1}{T}\sum_{k'=0}^{K-1} c_{k'} \vartheta_{k'}^{(l-k)}, \\
	&= \frac{1}{T} \sum_{k'=0}^{K-1} c_{k'} \vartheta_{k'}^{l}  \underbrace{\sum_{k=0}^{K} h_{k} \vartheta_{k'}^{-k}}_{H(\vartheta_{k'})=0} = 0, \; \forall l\in\bb Z.
\end{split}
\label{eq:annihilation}
\end{equation}
Let $\hat{\bd x} = [\hat{x}_{-M}\; \hat{x}_{-M+1} \; \cdots \; \hat{x}_{M-1}\; \hat{x}_M]^\TT \in \bb C^{2M+1}, \; M\geq K$ denote the $(2M+1)$ Fourier coefficients. Then, $\hat{\bd x}$ can be embedded into a Toeplitz matrix in $\bb C^{(M+1)\times (M+1)}$, using the Toeplitzification operator $\Gamma_M$ \cite{simeoni2020cpgd} constructed as follows:
\begin{align*}
	&\Gamma_M: \bb C^{2M+1} \rightarrow \bb C^{(M+1)\times (M+1)}, \\
	&\hat{\bd x} \mapsto (\Gamma_M\hat{\bd x})_{i,j} = \hat{x}_{i-j},\; i,j=1,\cdots,M+1.
\end{align*}
Considering $l \in \llbracket 0, M \rrbracket \overset{\text{def.}}{=} \{0,1,2,\ldots,M\}$ from Eq.~\eqref{eq:annihilation} gives the linear system of equations $(\Gamma_M\hat{\bd x})\bd h = \bld 0$, {\it i.e.}, the annihilating filter is a nontrivial vector in the nullspace of $\Gamma_M\hat{\bd x}$ and can be found using the Eckart-Young theorem \cite{horn2012matrix}, which selects the right eigenvector corresponding to the smallest singular value of the matrix $\Gamma_M\hat{\bd x}$. The locations $\{\tau_{k}\}_{k=0}^{K-1}$ can be determined from the roots $\{\vartheta_k\}_{k=0}^{K-1}$ as $\tau_k = -\frac{T}{2\pi}\angle{\vartheta_k}$, where $\angle$ denotes the principal angle of the complex argument.\\
\indent Robust techniques have been proposed for parameter recovery of FRI signals in the presence of noise. The Cadzow denoising algorithm \cite{cadzow,blu2008cadzow} is used to denoise the Fourier coefficients in the presence of measurement noise. Condat and Hiribayashi \cite{condat2015cadzow} proposed a convex optimization problem based on the \emph{lifting technique} to estimate the support parameters. Recovery of Fourier coefficients from the signal measurements can be posed as an $\ell_2$-minimization problem with a suitable regularizer to satisfy the annihilation condition. Do{\v{g}}an {\it et al.} \cite{dougan2015reconstruction} and Pan {\it et al.} \cite{pan2016towards} solved the problem using alternating minimization and Simeoni {\it et al.} \cite{simeoni2020cpgd} solved the problem using an approximate version of the proximal gradient method.

\subsection{Coefficient Recovery}
Once the support parameters are recovered, the coefficients $\bd a = [a_0\;a_1\cdots a_{K-1}]^\TT$ can be estimated using $\hat{\bd x}$ and the estimated support parameters $\bld \tau$ using linear least-squares regression on $\bd a$, {\it i.e.}, by solving for $\bd a$ in the following:
\begin{equation}\label{eq:amplitude_estimation}
	\hat{\bd x}\!=\!\underbrace{\begin{bmatrix}
	e^{-\mathrm{j}M\omega_0 \tau_{0}} & \hspace{-0.3cm}e^{-\mathrm{j}M\omega_0 \tau_{1}} & \hspace{-0.3cm}\cdots & \hspace{-0.3cm}e^{-\mathrm{j}M\omega_0 \tau_{K\!-\!1}}\\
    e^{-\mathrm{j}(M\!-\!1)\omega_0 \tau_{0}} & \hspace{-0.3cm} e^{-\mathrm{j}(M\!-\!1)\omega_0 \tau_{1}} & \hspace{-0.3cm}\cdots & \hspace{-0.3cm}e^{-\mathrm{j}(M\!-\!1)\omega_0 \tau_{K\!-\!1}}\\
	\vdots & \hspace{-0.3cm}\vdots & \hspace{-0.3cm}\ddots & \hspace{-0.3cm}\vdots\\
    e^{\mathrm{j}M\omega_0 \tau_{0}} & \hspace{-0.3cm} e^{\mathrm{j}M\omega_0 \tau_{1}} & \hspace{-0.3cm}\cdots & \hspace{-0.3cm}e^{\mathrm{j}M\omega_0 \tau_{K\!-\!1}}
	\end{bmatrix}}_{\bd V}\bd a.
\end{equation}
Since $\bd V\in \bb C^{(2M+1)\times K}$ has a Vandermonde structure, it is left-invertible \cite{horn2012matrix}, and the estimated coefficient vector $\bd a$ is unique.


\section{Kernel-based Neuromorphic Sampling}
\label{sec:siso_sampling}
In this section, we propose kernel-based neuromorphic sampling akin to kernel-based uniform sampling \cite{vetterli2002sampling,mulleti2017paley,haro2018sampling}, and develop a perfect reconstruction strategy for FRI signals from events generated using a neuromorphic encoder. A schematic of kernel-based neuromorphic sampling is shown in Figure~\ref{fig:kernel_based_sampling}.


\subsection{Reconstruction of Pulse Streams}
\label{subsec:pulse_stream_reconstruction}
Consider a $K$-sparse, compactly-supported signal $x(t)\in L_2([0,T])$ constructed using a prototype pulse $\varphi(t)\in L_2(\bb R)$, defined as
\begin{equation}\label{eq:pulse_signal_model}
	x(t) = \sum_{k=0}^{K-1} a_k\varphi(t-\tau_k),
\end{equation}
where $\bd a = [a_0\;a_1\cdots a_{K-1}]^\TT \in \bb R^K$ is the vector of unknown coefficients, $\bld \tau = [\tau_0\;\tau_1\;\cdots\;\tau_{K-1}]^\TT\in\bb R^K$ is the vector of unknown locations with $0\leq \tau_{0}<\cdots<\tau_{K-1}< T$. The signal $x(t)$ has a rate of innovation equal to $\displaystyle\frac{2K}{T}$. Signal reconstruction via parameter estimation requires two steps --- estimation of the support parameters using Prony's method, followed by estimation of the coefficients using least-squares regression. Since $x(t) \in L_2([0,T])$, it admits a Fourier series representation with Fourier coefficients given by
\begin{equation}\label{eq:swce_pulse}
\hat{x}_{l} = \frac{1}{T}\hat{\varphi}(l\omega_{0})\sum_{k=0}^{K-1} a_{k} e^{-\jj l\omega_{0}\tau_{k}}, \; l\in\bb Z,
\end{equation}
where $\hat{\varphi}$ denotes the Fourier transform of the pulse $\varphi$ and $\omega_0 = \frac{2\pi}{T}$. The Fourier coefficients possess the sum-of-weighted-complex-exponential (SWCE) structure, akin to Eq.~\eqref{eq:swce_dirac}. The support parameters can be estimated with a minimum of $2K+1$ Fourier coefficients using Prony's method (cf. Section~\ref{subsec:prony}).\\
\indent We rely on the kernel-based sampling strategy proposed in \cite{mulleti2017paley} to obtain the Fourier coefficients of the signal, {\it i.e.}, the signal is observed using a neuromorphic encoder after filtering through a sampling kernel $g(t)$ as shown in Figure~\ref{fig:kernel_based_sampling}. The sampling kernel is typically a lowpass filter that introduces additional smoothness in the signal that is input to the neuromorphic encoder $\sT_C$. The filtered signal is given by
\begin{equation} \label{eq:filteredSignalDerivation}
\begin{split}
	f(t) &= (x*g)(t)
	= \int_{\bb R} g(\nu) x(t-\nu) \,\dd \nu, \\
	&= \int_{\bb R} g(\nu) \sum_{l\in\bb Z} \hat{x}_l e^{\jj \omega_{0}l (t-\nu)}\, \dd \nu, \\
	&= \sum_{l\in\bb Z} \hat{x}_l \left( \int_{\bb R}g(\nu) e^{-\jj l\omega_{0} \nu}\, \dd \nu\right) e^{\jj \omega_{0}l t}, \\
	&= \sum_{l\in\bb Z} \hat{x}_l \hat{g}(l\omega_{0})e^{\jj \omega_{0}l t},
\end{split}
\end{equation}
where $\hat{g}(\omega)$ denotes the Fourier transform of $g(t)$. We see that the Fourier representation also extends to the filtered signal $f(t)$ with the coefficients given by $\hat{x}_l\,\hat{g}(l\omega_0)$. The advantage of the kernel-based sampling approach is that the kernel can be chosen to satisfy certain conditions, for instance, the Fourier-domain alias-cancellation conditions \cite{mulleti2017paley}:
\begin{equation} \label{eq:alias_cancellation}
\hat{g}(l\omega_{0}) = \begin{cases}
g_{l} \neq 0, & l \in \llbracket -K, K \rrbracket, \\
0, & l \notin \llbracket -K, K \rrbracket. \\
\end{cases}
\end{equation}
\begin{figure}[!t]
	\centering
	\includegraphics[width=3.5in]{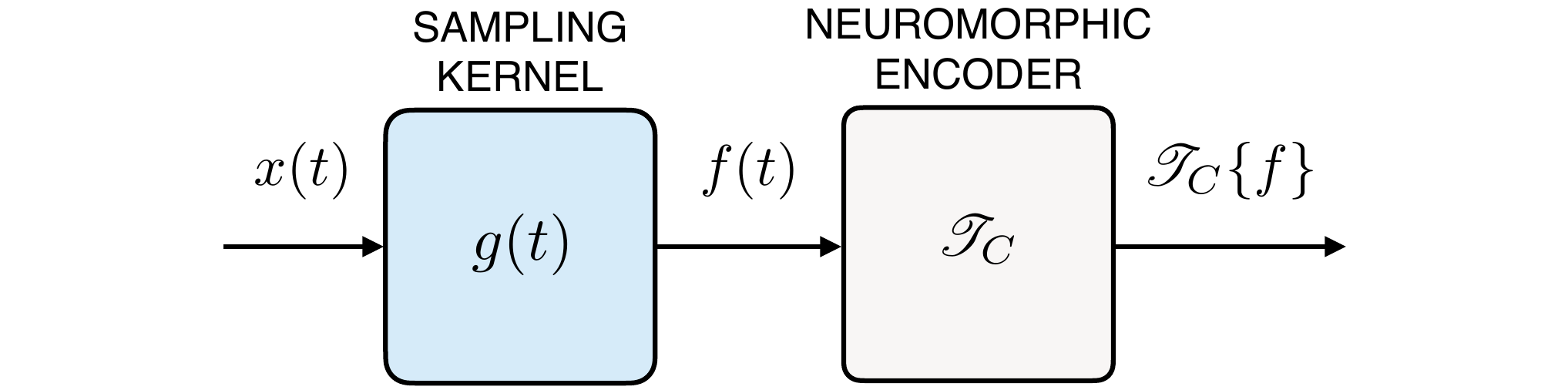}
	\caption{Schematic of kernel-based neuromorphic sampling of an FRI signal $x(t)$ using a sampling kernel $g(t)$ that satisfies the Fourier-domain alias-cancellation conditions in Eq.~\eqref{eq:alias_cancellation}.}
	\label{fig:kernel_based_sampling}
\end{figure}
\begin{figure*}[!t]
	\centering
	\subfigure[]{\includegraphics[width=2.1in]{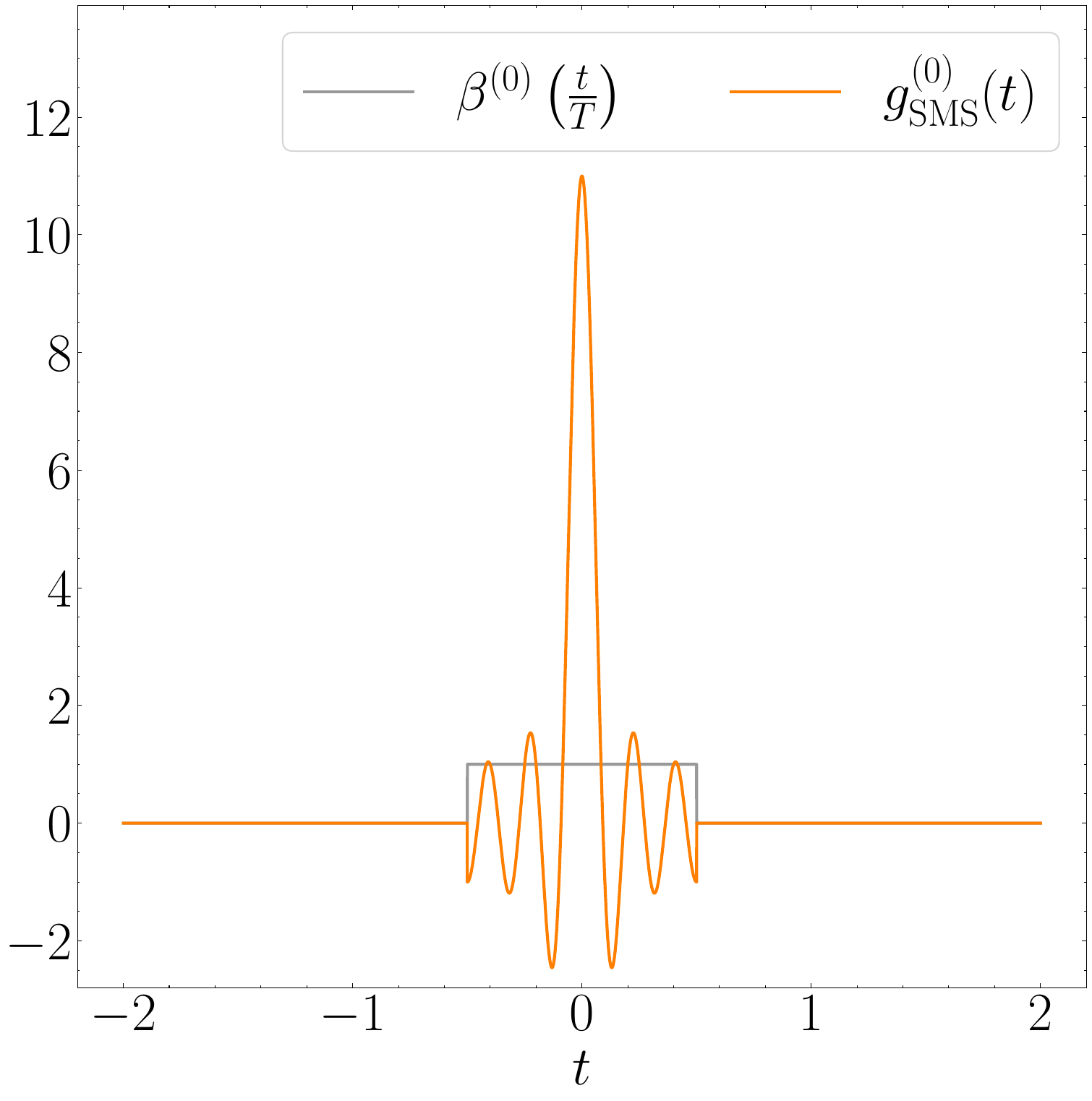}}
	\subfigure[]{\includegraphics[width=2.1in]{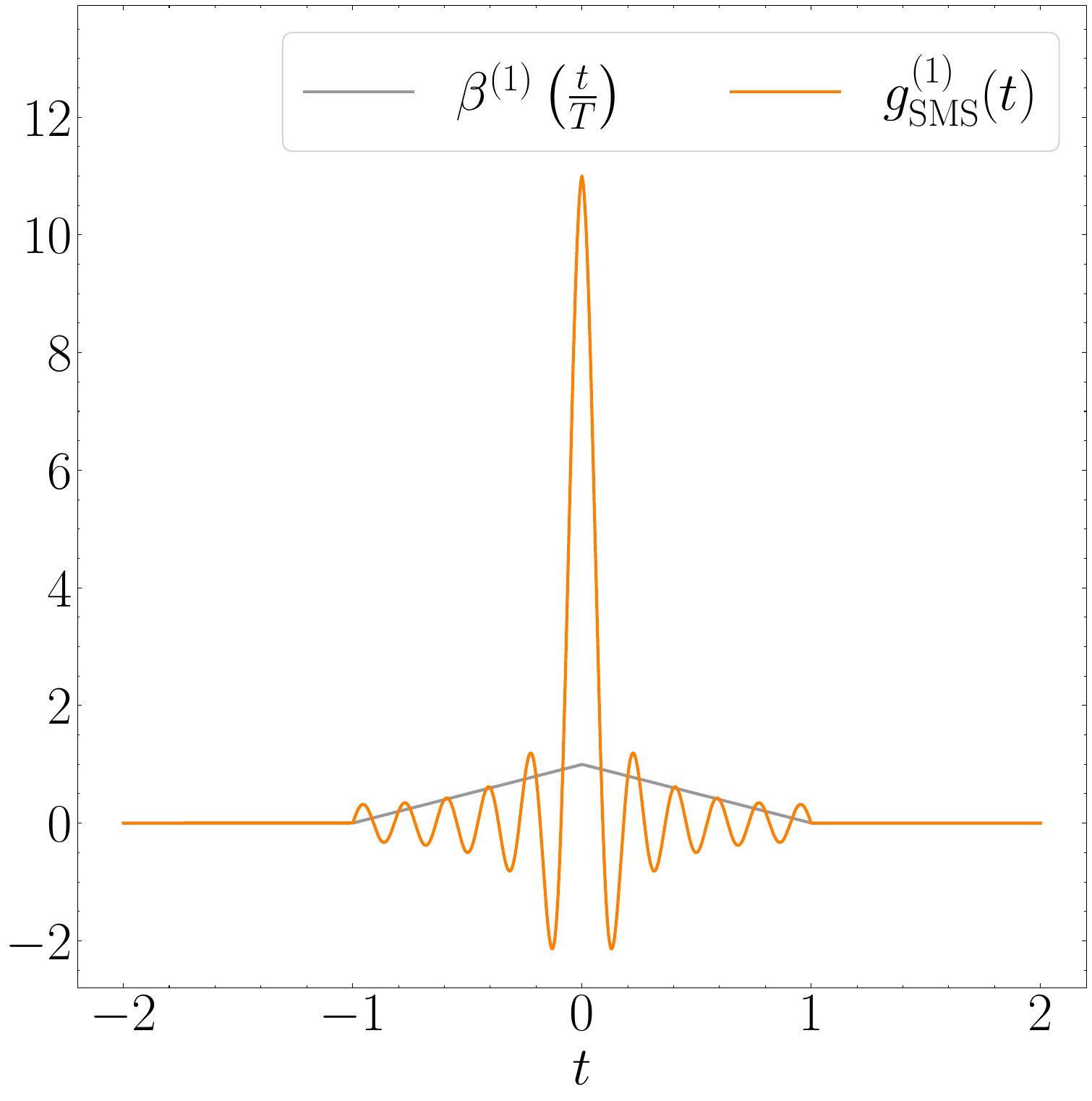}}
	\subfigure[]{\includegraphics[width=2.1in]{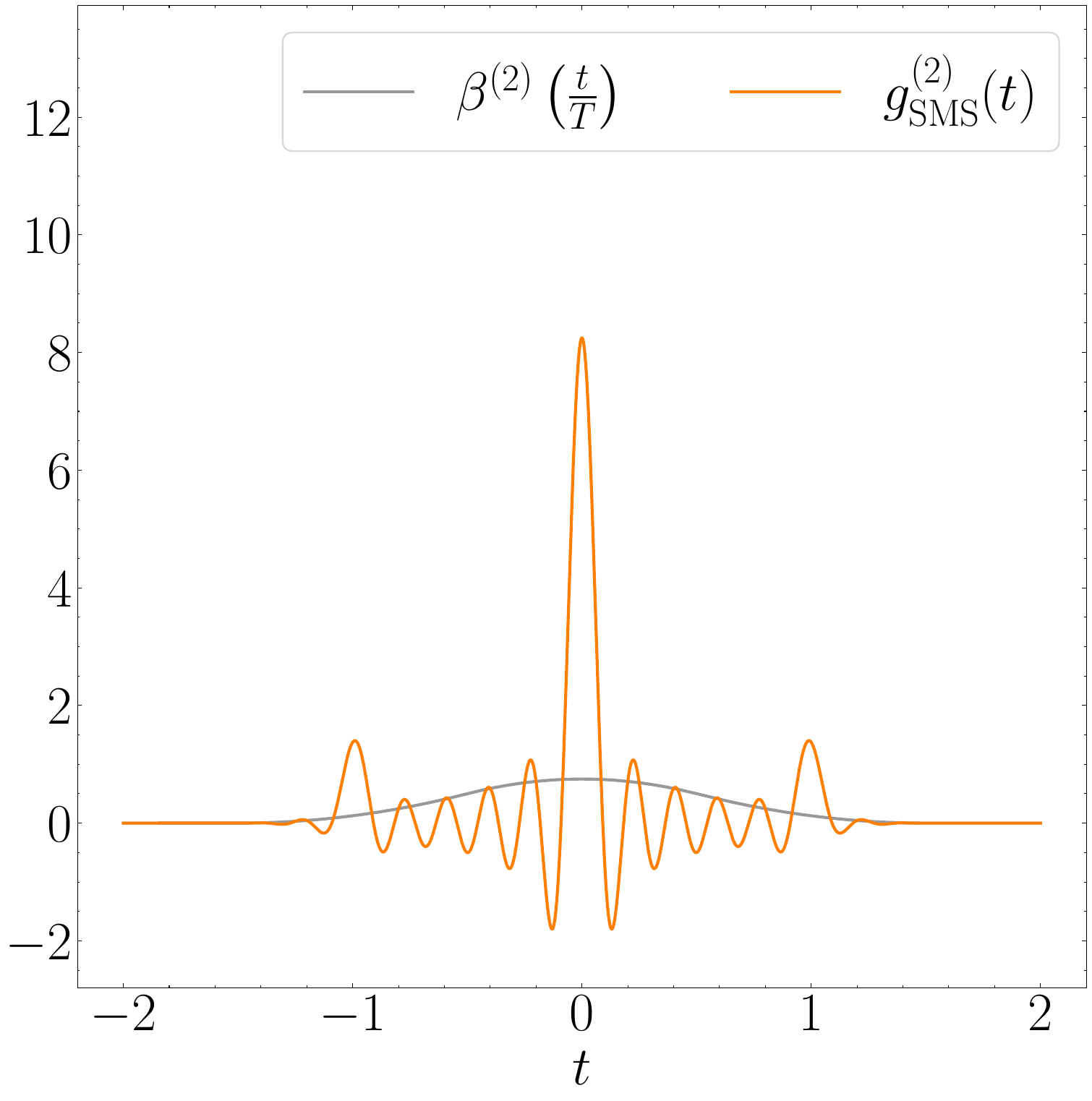}}
	\subfigure[]{\includegraphics[width=2.1in]{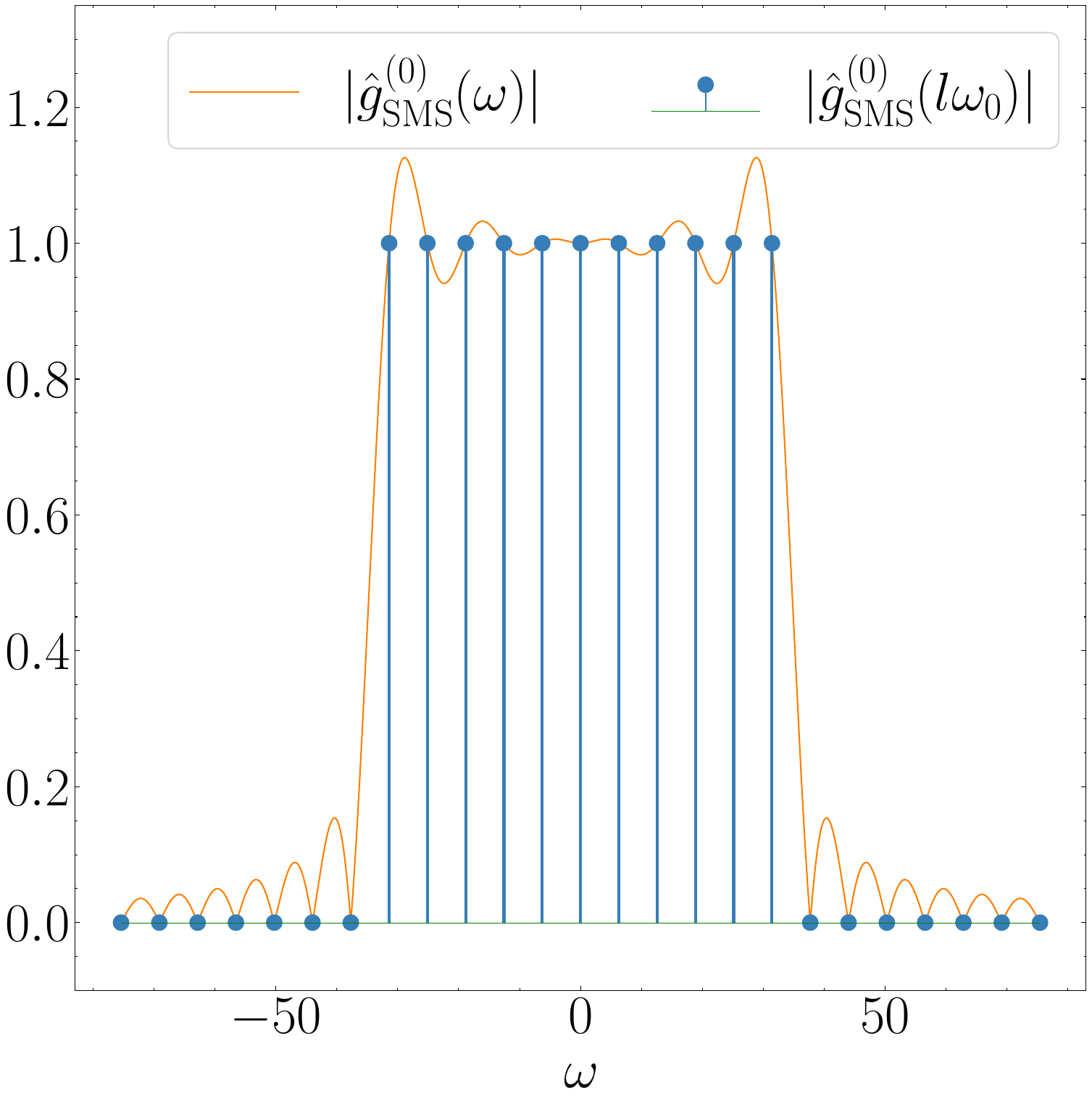}}
	\subfigure[]{\includegraphics[width=2.1in]{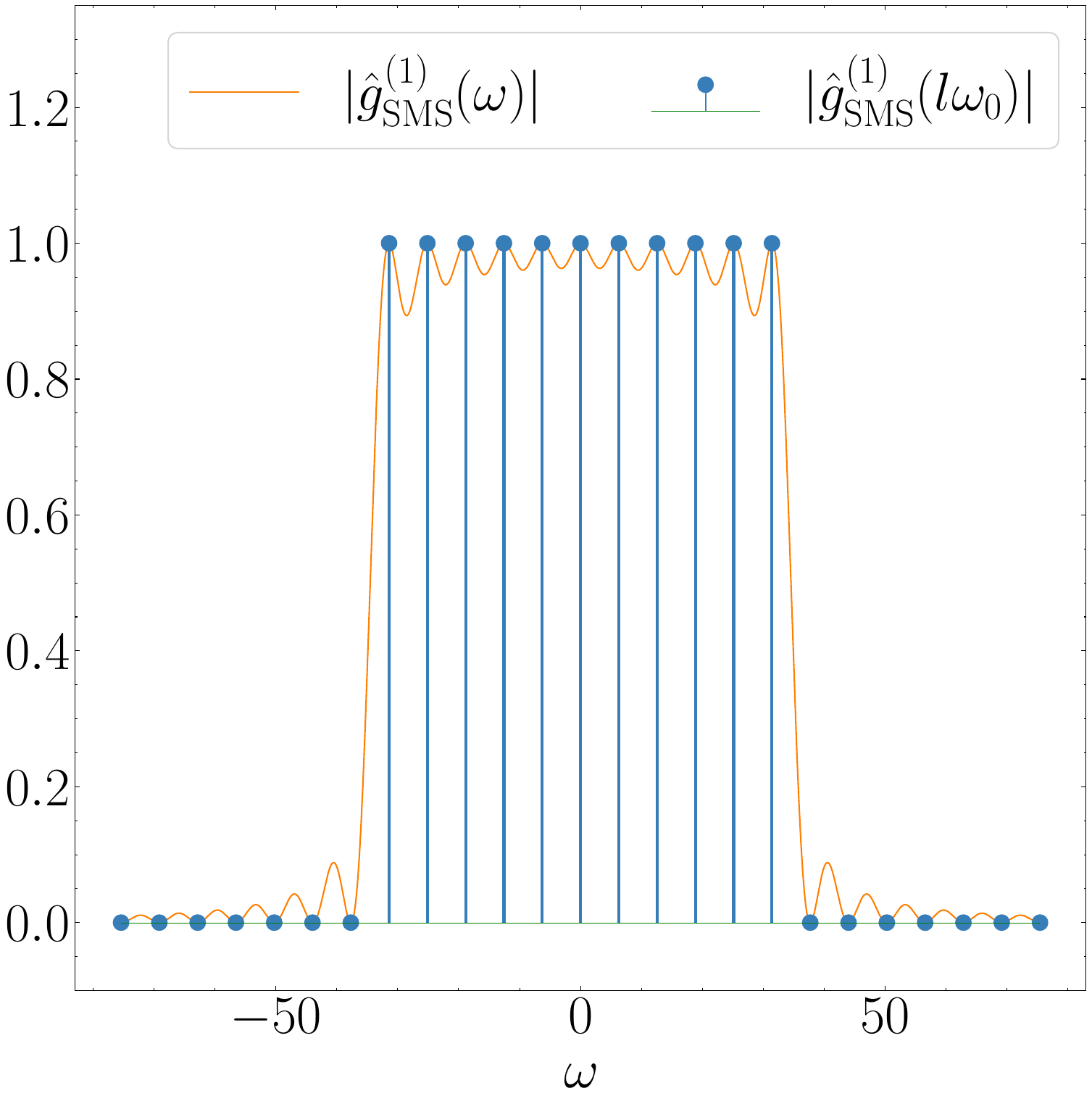}}
	\subfigure[]{\includegraphics[width=2.1in]{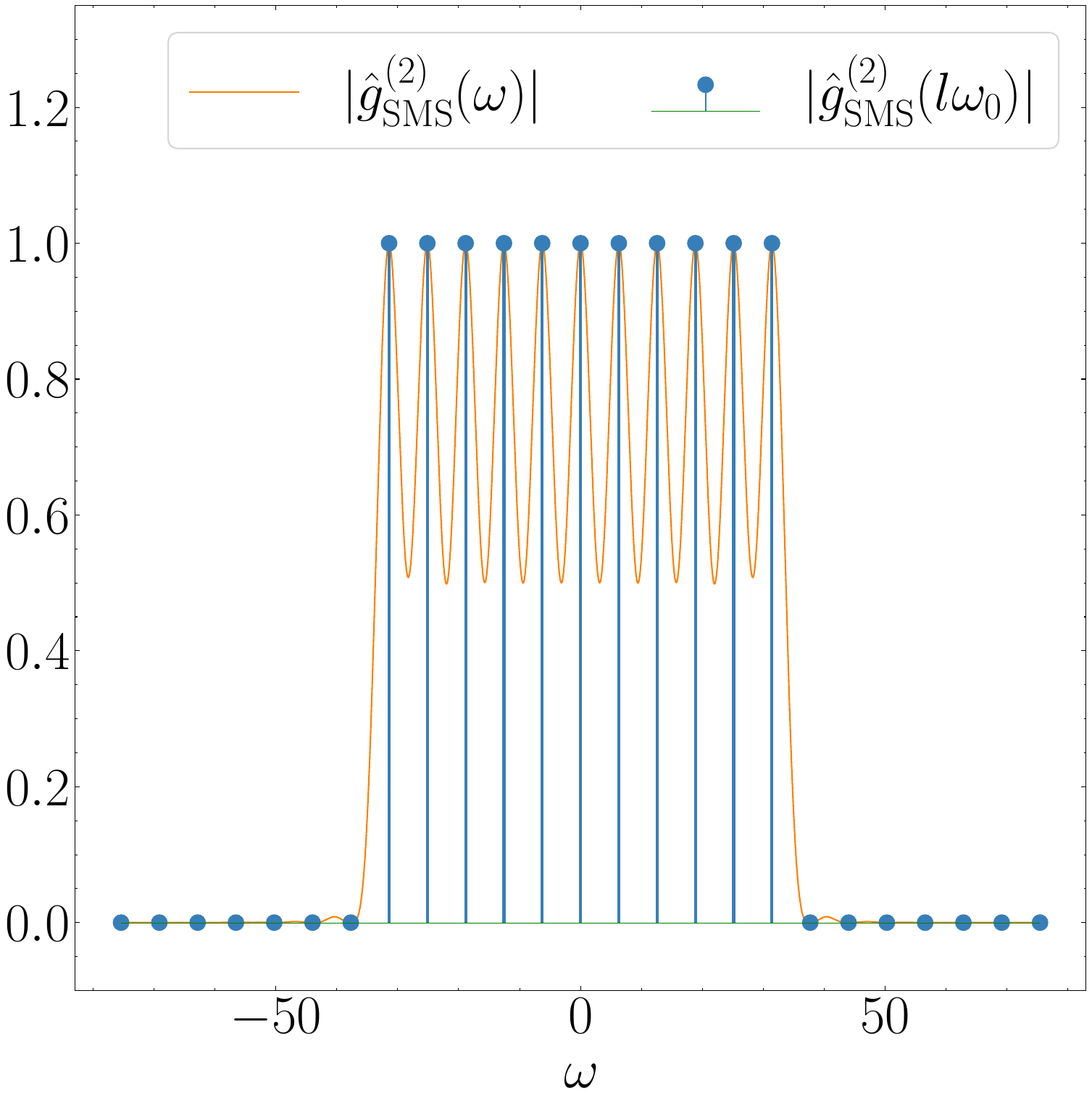}}
	\caption{The sum-of-modulated-splines (SMS) sampling kernels for $K=5$, and $T=1$: (a)-(c) depict the time-domain representations; and (d)-(f) depict the Fourier-domain representations, for $r=0,1,2$. The time-domain representations show that the SMS kernels are compactly supported. The Fourier-domain representations show that the kernels satisfy the Fourier-domain alias-cancellation conditions given in Eq.~\eqref{eq:alias_cancellation} with $\hat{g}^{(r)}_\text{SMS}(l\omega_0)=1, \forall l\in\llbracket -K,K\rrbracket$.}
	\label{fig:sampling_kernels}
\end{figure*}
Without loss of generality, setting $g_l=1$, gives
\begin{equation} \label{eq:trig_poly}
	f(t) = \sum_{l=-K}^K \hat{x}_l e^{\jj \omega_{0}l t},
\end{equation}
which is a finite sum of complex exponentials and has the desired continuity property for neuromorphic encoding. Examples of kernels that satisfy the Fourier-domain alias-cancellation conditions include the sinc function of bandwidth $\displaystyle\frac{2K+1}{T}$; the sum-of-sincs (SoS) kernel in the Fourier domain \cite{tur2011innovation}; the sum-of-modulated-spline (SMS) kernels in the time domain \cite{mulleti2017paley}; and exponential-reproducing and polynomial-reproducing kernels, which satisfy the {\it generalized Strang-Fix} conditions \cite{dragotti2007fix}. We use the SMS family of kernels \cite{mulleti2017paley}, which have the time-domain representation:
\begin{equation}\label{eq:sms_sampling_kernel}
	g^{(r)}_{\text{SMS}}(t) = \frac{1}{T} \sum_{k=-K}^K \beta^{(r)}\left(\frac{t}{T}\right)e^{\jj k\omega_0 t},
\end{equation}
where $\beta^{(r)}$ denotes the centered polynomial B-spline of degree $r$ \cite{unser1999splines}. In particular, for $r=0$, the SMS and the SoS kernels coincide, {\it i.e.}, $g_\text{SoS} = g_\text{SMS}^{(0)}$. SMS kernels are compactly supported, which makes them better suited for practical applications. Figure~\ref{fig:sampling_kernels} shows the SMS kernels, of the zeroth, first, and second order, in the time and frequency domains.\\
\indent Suppose there are $L$ events recorded in the interval $[0,T]$ comprising event instants $\{t_1, t_2, \cdots, t_L\}$, and corresponding polarities $\{p_1, p_2, \cdots, p_L\}$. The corresponding signal amplitudes $\bd f = [f(t_1)\; f(t_2) \cdots f(t_L)]^\TT$ follow Lemma~\ref{lem:ttransform}. From the preceding Fourier analysis, we can set up a linear system of equations relating the $2K+1$ Fourier coefficients and the signal amplitudes as $\bd f = \bd G\hat{\bd x}$, where the matrix $\bd G\in\bb C^{L\times (2K+1)}$ is given as follows:
\begin{equation} \label{eq:ctemMatrix}
	\bd G = \begin{bmatrix}
    e^{-\mathrm{j}K\omega_0 t_{1}}  & \cdots & 1  &\cdots &e^{\mathrm{j}K\omega_0 t_{1}}\\
    e^{-\mathrm{j}K\omega_0 t_{2}}  & \cdots & 1 &\cdots &e^{\mathrm{j}K\omega_0 t_{2}}\\
    \vdots  & \ddots & \vdots & \ddots &\vdots\\
    e^{-\mathrm{j}K\omega_0 t_{L}}  & \cdots & 1 &\cdots &e^{\mathrm{j}K\omega_0 t_{L}}  \/
    \end{bmatrix}.
\end{equation}
Solving for $\hat{\bd x}$ from $\bd f = \bd G\hat{\bd x}$ gives the Fourier coefficients for $l\in \llbracket -K, K \rrbracket$ that follow Eq.~\eqref{eq:swce_pulse}. Thereafter, the vector of support parameters $\bld \tau$ can be estimated using Prony's annihilating filter method and the vector of amplitude parameters $\bd a$ can be estimated using least-squares regression. We show that it is indeed possible to solve for $\hat{\bd x}$ uniquely in ${\bd f} = \bd G\hat{\bd x}$ by virtue of the left-invertibility of $\bd G$ for $L\geq 2K+1$, which is established next.
\begin{lemma} \label{lem:ctemMatrix}
	The matrix $\bd G \in \bb C^{L\times (2K+1)}$ defined in Eq.~\eqref{eq:ctemMatrix} has full column-rank whenever $L \geq (2K+1)$.
\end{lemma}
\begin{proof}
	Consider the case where $L=2K+1$, {\it i.e.}, $\bd G$ is a square matrix. By analyzing the determinant, we show that $\bd G$ has full rank. This property carries over to the case when $\bd G$ is a tall matrix ($L\geq 2K+1$), because addition of rows to the $(2K+1)\times (2K+1)$ matrix cannot decrease the rank.\\
	\indent Using the properties of determinants:
	\begin{equation} \label{eq:prodGct}
	\det(\bd G) = \det(\bd F)\cdot \prod_{m=1}^L e^{-\jj K \omega_0 t_m},
	\end{equation}
	where
	\begin{equation*}
	\bd F =
	\begin{bmatrix}
	1  & e^{\mathrm{j}\omega_0 t_{1}} &\hspace{-0.2cm}\cdots &\hspace{-0.1cm}e^{\mathrm{j}2K\omega_0 t_{1}}\\
	1 & e^{\mathrm{j}\omega_0 t_{2}} &\hspace{-0.2cm}\cdots &\hspace{-0.1cm}e^{\mathrm{j}2K\omega_0 t_{2}}\\
	\vdots & \vdots & \ddots & \vdots\\
	1 & e^{\mathrm{j}\omega_0 t_{L}} &\hspace{-0.2cm}\cdots &\hspace{-0.1cm}e^{\mathrm{j}2K\omega_0 t_{L}}  \/
	\end{bmatrix}
	\end{equation*}
is a Vandermonde matrix. The set $\{t_m\}_{m=1}^L$ has increasing and distinct entries, {\it i.e.},  $0 < t_1 < t_2 < \cdots < t_L$, and $\omega_0 = \frac{2\pi}{T}$ with $T > t_L$, which makes the rows of $\bd F$ distinct. Using the properties of Vandermonde matrices \cite{horn2012matrix}, we know that $\det(\bd F) \neq 0$, whenever $L \geq 2K+1$, and each term in the product in Eq.~\eqref{eq:prodGct} is nonzero. Hence, $\det(\bd G) \neq 0$, whenever $L \geq 2K+1$.
\end{proof}
Once the Fourier coefficients are estimated, the $(K+1)$-tap annihilating filter can be determined using Prony's method (cf. Section~\ref{subsec:prony}),
the roots of which are in one-to-one correspondence with the locations $\bld \tau$. Subsequently, the coefficient vector $\bd a$ can be obtained by solving the following system of equations (cf. Eq.~\eqref{eq:swce_pulse}) $\hat{\bd x} = \bd S\bd V\bd a$,
where $\bd S = \text{diag}\left\{\hat{\varphi}(l\omega_0)\right\}_{l=-M}^M$ is a diagonal matrix, and is invertible. $\bd V\in \bb C^{(2M+1)\times K}$, where $2K+1\leq 2M+1\leq L$, has a Vandermonde structure \cite{horn2012matrix} with distinct entries and is invertible. Therefore, $\bd S\bd V$ is invertible.\\
\indent The minimum sampling requirement of $2K+1$ measurements for perfect reconstruction can be ensured by setting the temporal contrast threshold sufficiently low compared with the dynamic range of the filtered signal, according to Lemma~\ref{lem:min_measurements}. We summarize the preceding discussion and provide the sufficient condition for perfect reconstruction in the following proposition.
\begin{proposition}[Perfect reconstruction of a stream of pulses]\label{prop:sparse_bound}
The signal $x(t)$ in Eq.~\eqref{eq:pulse_signal_model} can be perfectly recovered from the events $\sT_C\{f\}$, where $f(t) = (x*g)(t)$, when $g(t)$ satisfies the alias-cancellation conditions (cf. Eq.~\eqref{eq:alias_cancellation}) and the temporal contrast threshold satisfies
\[
    0 < C < \frac{f_{\max} - f_{\min}}{2K+1},
\]
where $\displaystyle f_{\max} = \max_{t\in [0,T]} f(t)$ and $\displaystyle f_{\min} = \min_{t\in [0,T]} f(t)$.
\end{proposition}
The critical temporal contrast threshold for a signal $f(t)$ in a $\Delta$-shift-invariant space is defined as \cite{kamath2023neuromorphic}:
\[
	C_f(\Delta) \overset{\text{def.}}{=} \frac{1}{2} \inf_{\tau\in\bb R} \abs{\max_{\tau<t<\tau+\Delta} f(t) - \min_{\tau<t<\tau+\Delta} f(t)}.
\]
If the temporal contrast threshold $C < C_f(\Delta)$, then neuromorphic encoding is guaranteed to generate at least one event in an interval of size $\Delta$. We demonstrated perfect reconstruction of signals in shift-invariant spaces from neuromorphic measurements subject to this condition \cite{kamath2023neuromorphic}.\\
\indent In the FRI context, the critical threshold definition takes the form:
\[
	C_f(T) = \frac{1}{\rho+1} \abs{\max_{t \in [0,T]} f(t) - \min_{t \in [0,T]} f(t)},
\]
where $\rho$ is the rate of innovation of the FRI signal $x(t)$. For $x(t)$ considered in \eqref{eq:pulse_signal_model}, $\rho = 2K$, and hence the upper bound of the temporal contrast threshold in Proposition~\ref{prop:sparse_bound} is the critical threshold. The definition of the critical threshold is also consistent with signals in shift-invariant spaces, including bandlimited signals. For instance, signals bandlimited  to $[-\frac{\pi}{T}, \frac{\pi}{T}]$ have a rate of innovation of $\frac{1}{T}$. The critical threshold for bandlimited signals is half the dynamic range in each interval of size $T$, which translates to having at least one event/measurement in a $T$-length interval.\\
\indent Following the {\it time-amplitude duality} \cite{martinez2019delta} in neuromorphic sampling, reconstruction from events requires knowledge of the signal's dynamic range $[f_{\min}, f_{\max}]$, as opposed to knowledge of the sampling interval/density in the case of uniform/nonuniform sampling. Since FRI signals are inherently sparse, the upper bound on the temporal contrast threshold is realizable, which enables perfect reconstruction. The technique for perfect reconstruction of FRI signals from neuromorphic measurements is summarized in Algorithm~\ref{algo:sparse}.
\begin{figure*}[!t]
\setcounter{subfigure}{-2}
	\centering
	\subfigure{\label{fig:}\includegraphics[height=.215in]{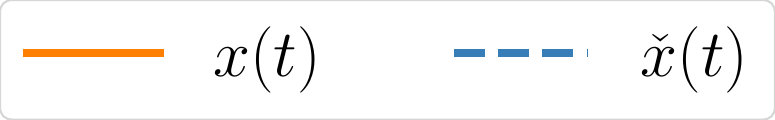}}
	\subfigure{\label{fig:}\includegraphics[height=.215in]{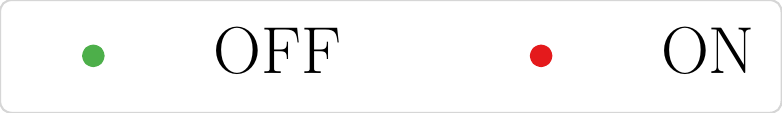}}

	\subfigure[Stream of Dirac impulses]{\label{fig:recons_sparse}\includegraphics[width=3.2in]{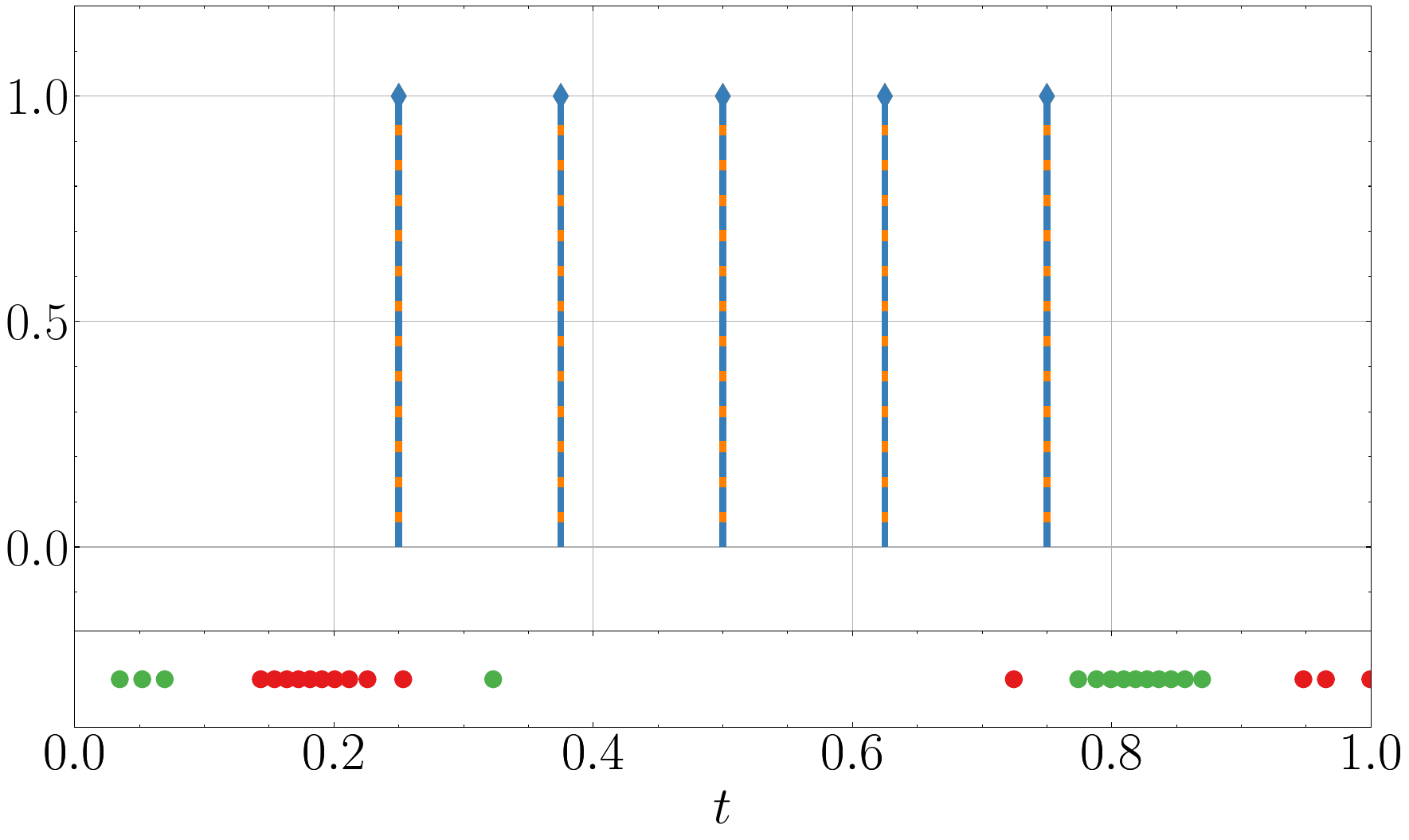}}
	\subfigure[Stream of pulses]{\label{fig:recons_pulse}\includegraphics[width=3.2in]{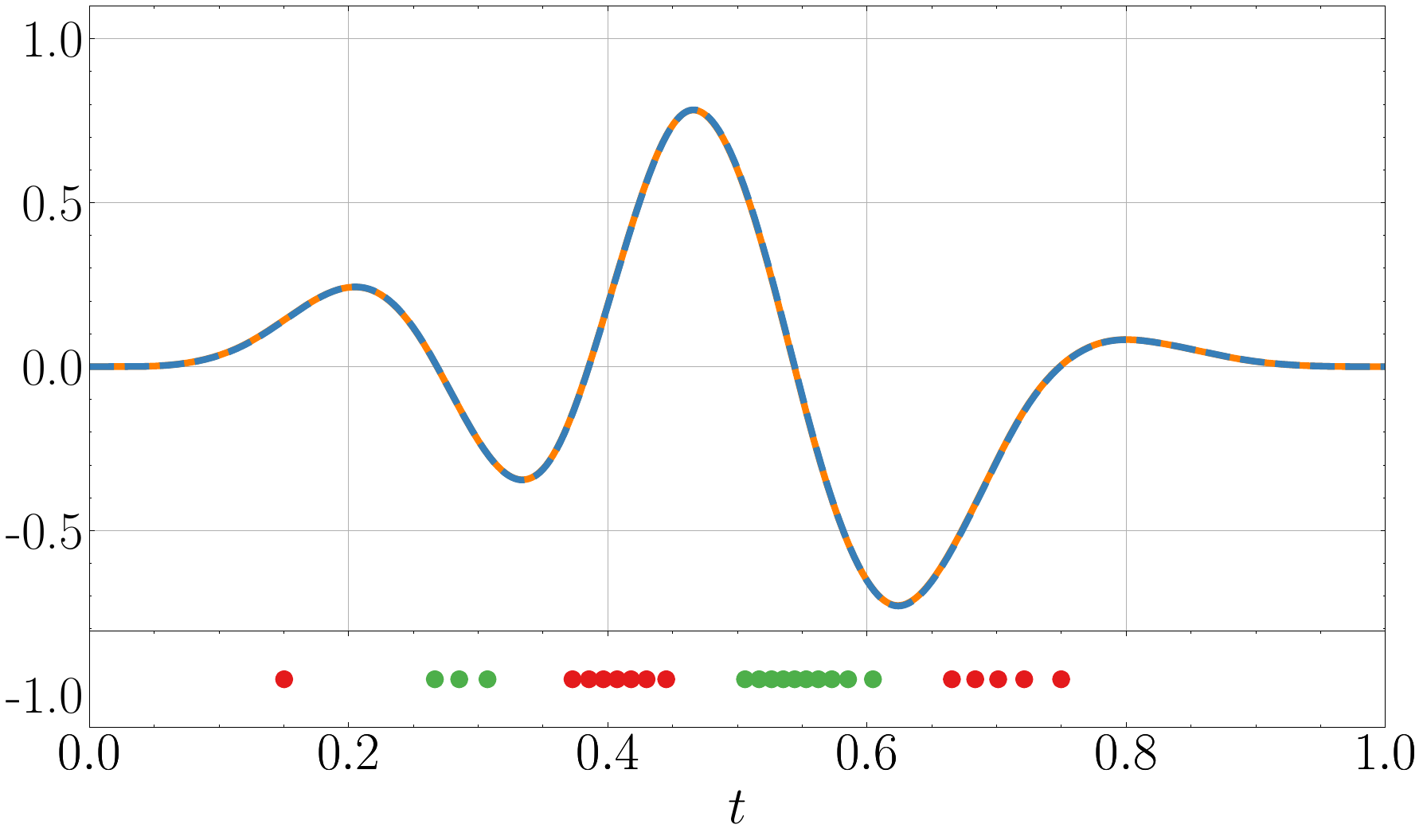}}
	\subfigure[A piecewise-constant signal]{\label{fig:recons_linear}\includegraphics[width=3.22in]{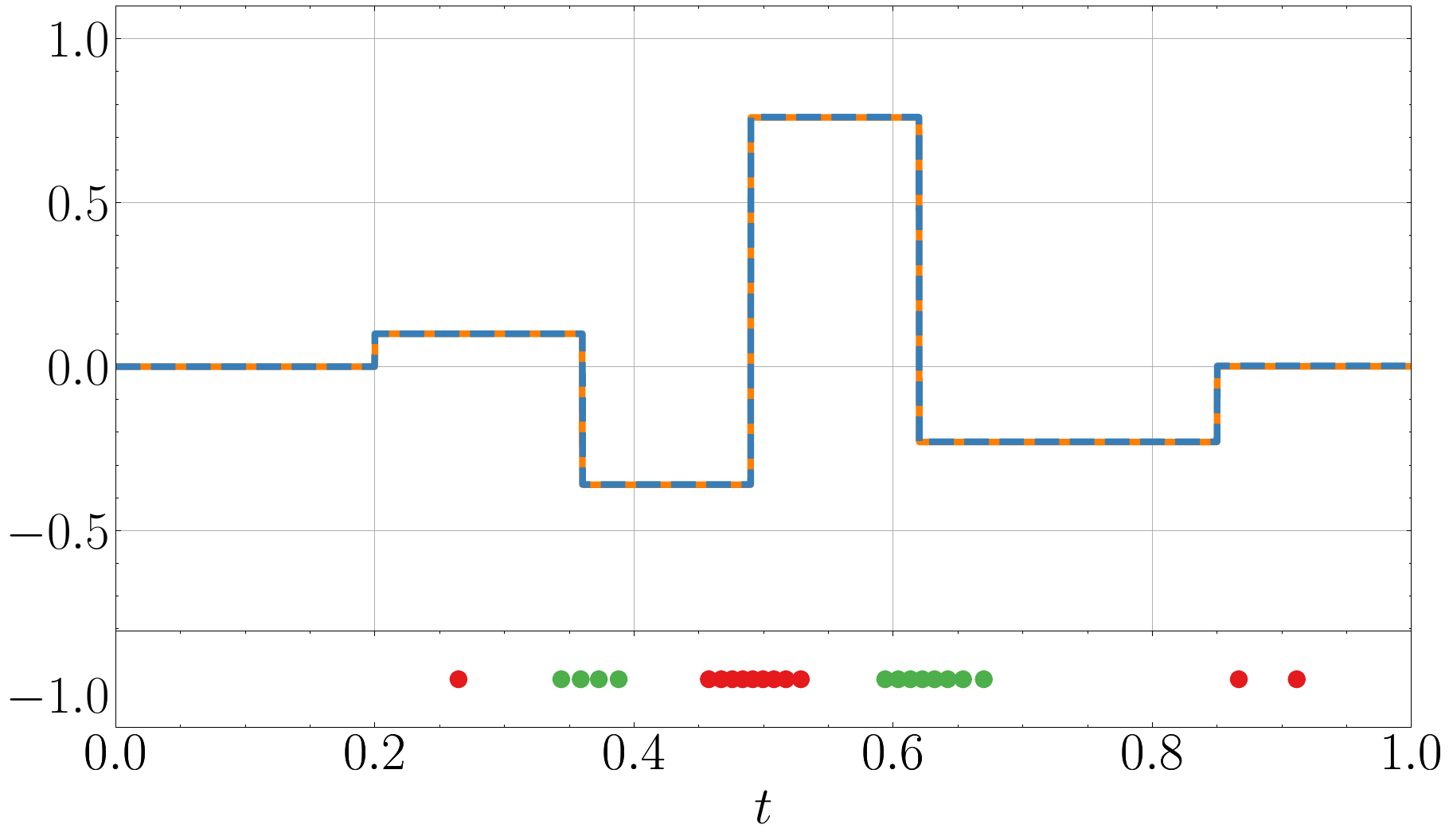}}
	\subfigure[A piecewise-linear signal]{\label{fig:recons_quad}\includegraphics[width=3.22in]{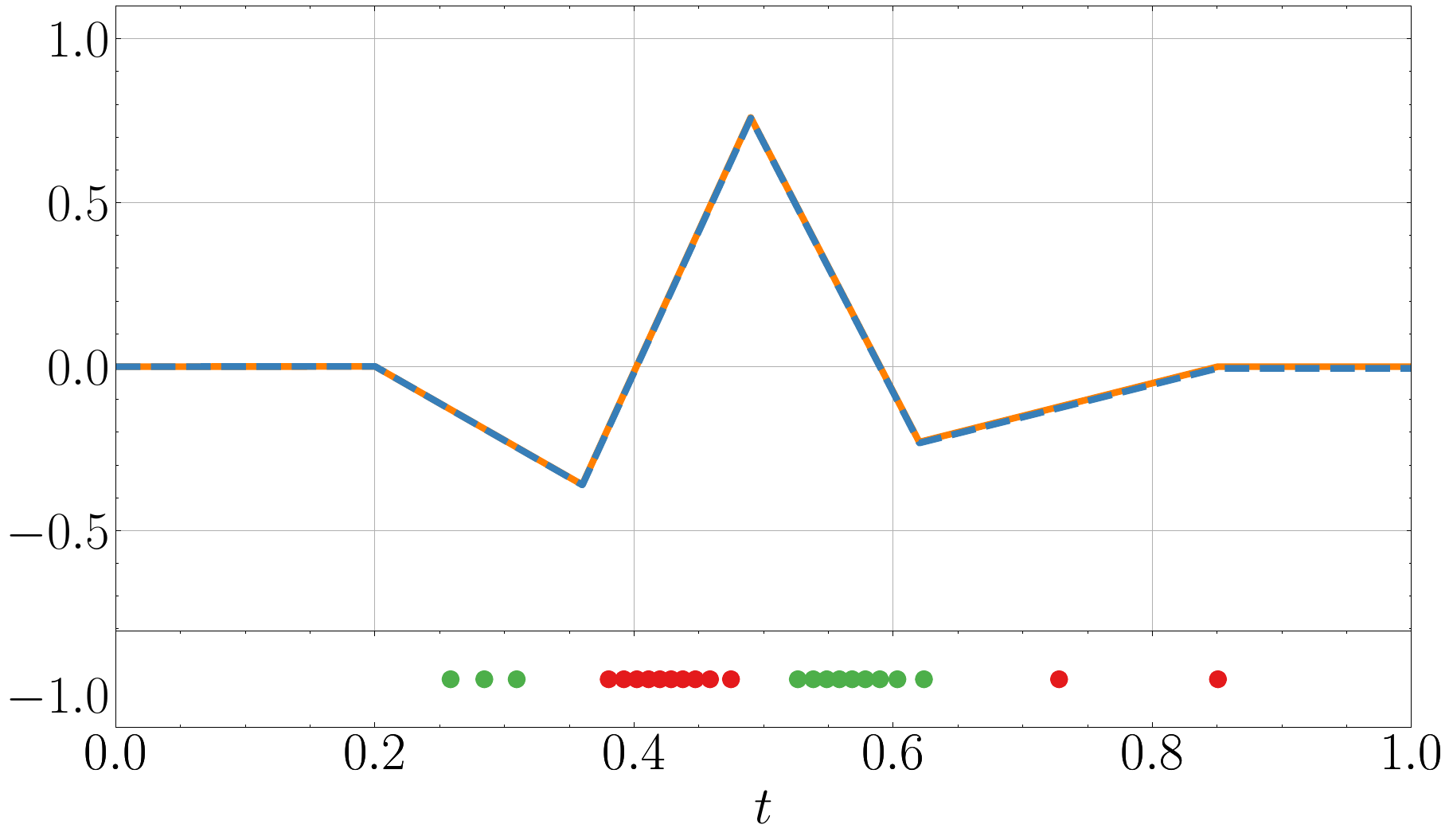}}
	\caption{Neuromorphic sampling of FRI signals using the zeroth-order SMS sampling kernel $g(t)$ and perfect reconstruction using Prony's method for various examples. The FRI signal $x(t)$, the reconstruction $\check{x}(t)$; and the corresponding ON and OFF events $\{(t_1,p_1), (t_2,p_2), \ldots\}$ are also shown.}
	\label{fig:SISO_reconstruction}
\end{figure*}

\subsection{Reconstruction of Nonuniform $\mathrm{L}$-Splines}
\label{subsec:nonuniform_spline_reconstruction}
A signal $x(t)$ is said to be a nonuniform $\mathrm L$-spline of degree $n\in\bb N$ with {\it knots} at $\{\tau_k\}_{k=0}^{K-1}$ if its $(n+1)$\textsuperscript{th} derivative is a stream of Dirac impulses \cite{unser1999splines}:
\begin{equation}\label{eq:nonuniform_Lsplines}
	\mathrm{D}^{n+1}\{x\}(t) = \sum_{k=0}^{K-1} a_k\delta(t-\tau_k).	
\end{equation}
Piecewise-constant and piecewise-polynomial functions are examples of $\mathrm{L}$-splines. The derivative property of the $\mathrm{L}$-spline makes a natural connection with continuous-time sparse/FRI signals, defined in Eq.~\eqref{eq:sparse_model}, under the operator $\mathrm{L} = \mathrm{D}^{n+1}$.\\
\indent Consider neuromorphic sampling of the nonuniform $\mathrm L$-spline $x(t)$. The Fourier coefficients of $\mathrm{D}^{n+1}\{x\}(t)$ are given by
\begin{equation}\label{eq:swce_spline}
	\hat{x}_l = \frac{1}{\left(\jj\omega_0 l\right)^{n+1}}\frac{1}{T}\sum_{k=0}^{K-1} a_{k} e^{-\jj \omega_{0}l \tau_{k}},
\end{equation}
which has the sum-of-complex-exponentials structure and is amenable to Prony's method (cf. Section~\ref{subsec:prony}) for support estimation. The analysis in Section~\ref{subsec:pulse_stream_reconstruction} carries over to the nonuniform $\mathrm{L}$-spline scenario with the exception that the diagonal matrix $\bd S$ is now given by $\bd S = \text{diag}\left\{ (\jj \omega_0 l)^{-n-1} \right\}_{l=-M}^M$.\\
\indent The minimum sampling requirement is $2K+1$ measurements, which gives identical bounds on the temporal contrast threshold as stated in Proposition~\ref{prop:sparse_bound}.
\begin{algorithm}[t]
    \caption{Perfect reconstruction of FRI signals from neuromorphic measurements.} \label{algo:sparse}
    \KwIn{Measurements $\sT_C\{f\} = \{(t_m, p_m)\}_{m=1}^L$, temporal contrast threshold $C$, model order $K$}
    {\bf Compute:} $\displaystyle f(t_m) = f(t_0) + C\sum_{i=1}^m p_i, \; m\in \llbracket 1,L\rrbracket$\;
    {\bf Solve Eq.~\eqref{eq:ctemMatrix}:} $\bd f = \bd G\hat{\bd x}$, for $\hat{\bd x}$ \;
    {\bf Support recovery:} Solve $(\Gamma_M\hat{\bd x})\bd h\! =\! \bld 0$ to obtain $\bld \tau$\;
    {\bf Coefficient recovery:} Solve $\hat{\bd x}\! =\! \bd S\bd V\bd a$ to obtain $\bd a$\;
\end{algorithm}


\subsection{Experimental Results}
We validate the proposed reconstruction strategy through simulations. Consider a stream of $K=5$ Dirac impulses
\[
	x(t) = \sum_{k=0}^4 \delta(t-\tau_k),	
\]
where $\{\tau_k\}_{k=0}^{4}=\{0.25, 0.375, 0.5, 0.625, 0.75\}$, {\it i.e.}, the locations are equally spaced in $[0.25,0.75]$ and the coefficients are set to unity. We use the zeroth-order SMS kernel \cite{mulleti2017paley} as the sampling kernel. The choice of the signal corresponds to the condition where the filtered signal $f(t)=(x*g)(t)$ has a low amplitude variation. The events $\sT_C\{f\}$ are obtained by setting the temporal contrast threshold $C=1/11$, according to Proposition~\ref{prop:sparse_bound}. The reconstruction is based on Algorithm~\ref{algo:sparse}. Figure~\ref{fig:recons_sparse} shows the stream of Dirac impulses $x(t)$, the filtered signal $f(t)$ and the events $\sT_C\{f\}$, and the reconstruction $\check{x}(t)$ with parameters estimated accurately up to machine precision, indicating perfect reconstruction for all practical purposes.\\
\indent Next, we consider a stream of $K=5$ pulses
\[
	x(t) = \sum_{k=0}^4 a_k\varphi(t-\tau_k),	
\]
where the coefficients are chosen as $\{a_k\}_{k=0}^{4}=\{0.49, -0.65,  0.47, -0.52,  0.22\}$, and support parameters are set to $\{\tau_k\}_{k=0}^{4}=\{0.22, 0.35, 0.46, 0.62, 0.79\}$, and the pulse $\varphi(t) = \beta^{(3)}\left(\frac{t}{10}\right)$ is a time-scaled cubic B-spline. We use the zeroth-order SMS kernel for sampling and set the temporal contrast threshold $C=0.015$ for obtaining events $\sT_C\{f\}$. Reconstruction is carried out following Algorithm~\ref{algo:sparse}. Figure~\ref{fig:recons_pulse} shows the stream of pulses $x(t)$, the filtered signal $f(t)$, and the recorded events $\sT_C\{f\}$, together with the reconstruction $\check{x}(t)$. The parameters are estimated accurately up to machine precision, indicating perfect reconstruction.\\
\indent Next, we consider nonuniform $\mathrm D^j$-splines, for $j=1, 2$. The choice $j=1$ corresponds to piecewise-constant signals and $j=2$ corresponds to piecewise-linear signals. The signals are acquired using the zeroth-order SMS kernel \cite{mulleti2017paley}, and the events are obtained by setting the temporal contrast threshold $C=$ for $j=1$, and $C=$ for $j=2$, following Proposition~\ref{prop:sparse_bound}. We use Prony's method as described in Algorithm~\ref{algo:sparse} for reconstruction. Figure~\ref{fig:recons_linear} and Figure~\ref{fig:recons_quad} show the $\mathrm D^j$-sparse signal $x(t)$ for $j=1,2$, the filtered signal $f(t)$ and the events $\sT_C\{f\}$ obtained, along with the reconstruction $\check{x}(t)$, respectively. Yet again, we have parameter estimation up to machine precision, indicating perfect reconstruction.


\begin{figure*}[!t]
	\centering
	\subfigure[Kernel-based SIMO neuromorphic sampling]{\label{fig:simo_encoding_schematic}\includegraphics[width=3.2in]{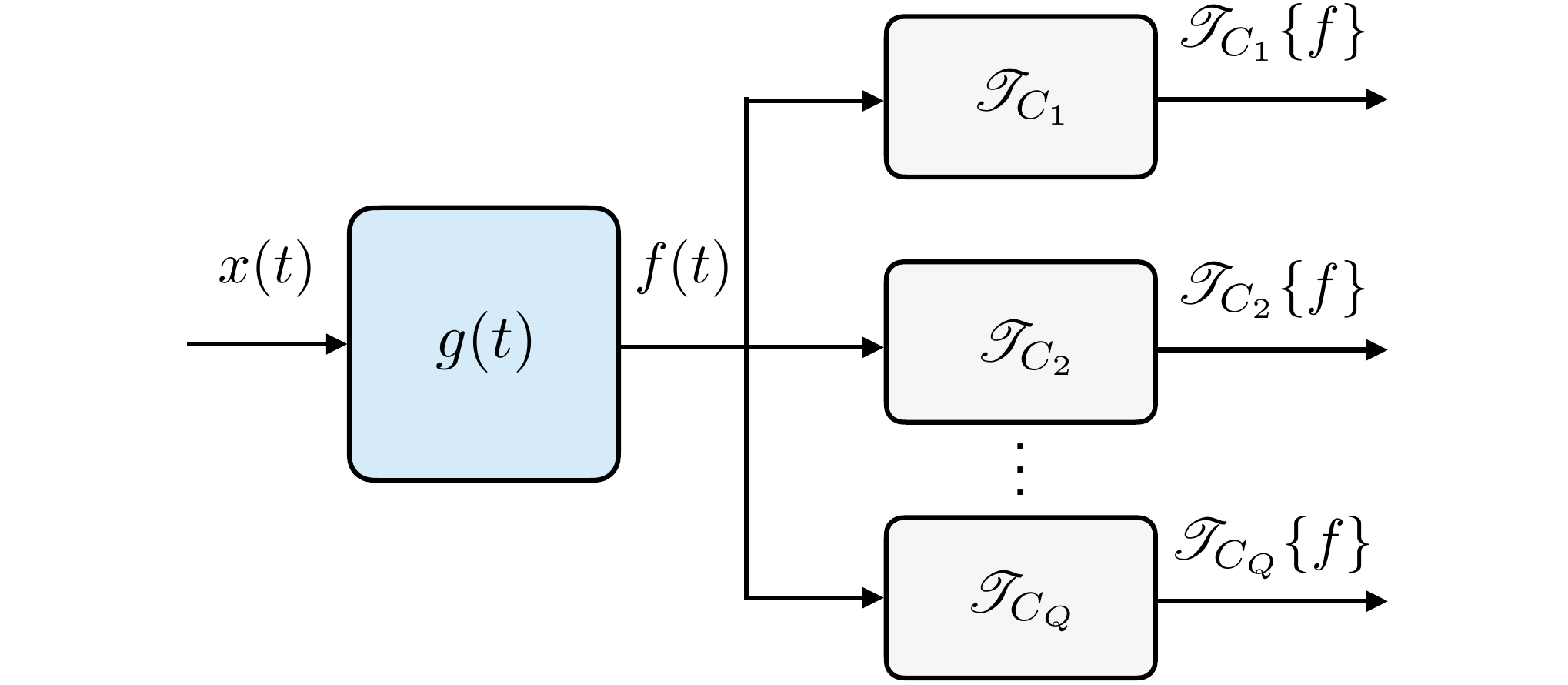}}
	\subfigure[Kernel-based MIMO neuromorphic sampling]{\label{fig:mimo_encoding_schematic}\includegraphics[width=3.2in]{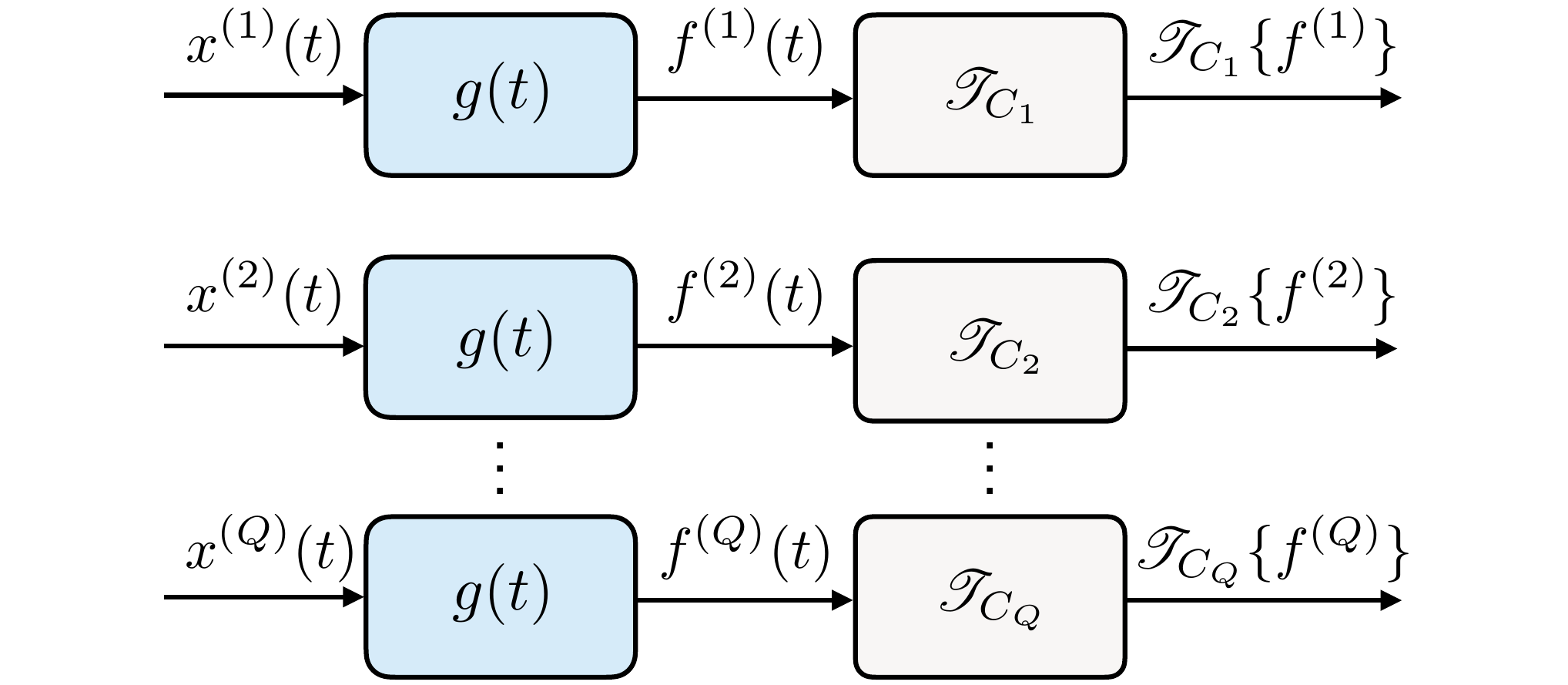}}
	\caption{Schematic of kernel-based multichannel neuromorphic sampling of FRI signals using a sampling kernel $g(t)$ that satisfies the alias-cancellation condition in Eq.~\eqref{eq:alias_cancellation}: (a) shows the single-input multi-output (SIMO) configuration; and (b) shows the multi-input multi-output (MIMO) configuration.}
	\label{fig:multichannel_schematic}
\end{figure*}
\section{Multichannel Neuromorphic Sampling}
\label{sec:multichannel_sampling}
We now consider multichannel extension of neuromorphic sampling of FRI signals in two configurations --- single-input multi-output (SIMO); and multi-input multi-output (MIMO). The corresponding schematics are shown in Figure~\ref{fig:multichannel_schematic}. The number of neuromorphic encoders is $Q$. Multichannel FRI signal models are encountered in applications such as ultrasound imaging \cite{tur2011innovation} and radar imaging \cite{rudresh2017radar}.


\subsection{SIMO Neuromorphic Sampling}
\label{subsec:simo_sampling}
Consider SIMO neuromorphic sampling (Figure~\ref{fig:simo_encoding_schematic}) of the FRI signal
in Eq.~\eqref{eq:pulse_signal_model}. Let $\sT_{C^{(i)}}\{f\} = \{(t^{(i)}_m,p^{(i)}_m)\}_{m=0}^{L^{(i)}}$ be the events generated by the $i$\textsuperscript{th} channel, where $f(t) = (x * g)(t)$ denotes the filtered FRI signal, and $L^{(i)}$ denotes the number of measurements obtained in the $i$\textsuperscript{th} channel. Using Lemma~\ref{lem:ttransform}, we obtain the amplitude samples of the vector that is input to the neuromorphic encoder: $\bd f^{(i)} = \left[f^{(i)}\left(t^{(i)}_1\right)\; f^{(i)}\left(t^{(i)}_2\right) \cdots f^{(i)}\left(t^{(i)}_{L^{(i)}}\right)\right]^\TT$, which are linearly related to the Fourier coefficients of the FRI signal. In vector notation, we can write $\bd f^{(i)} = \bd G^{(i)}\hat{\bd x}, \; i\in\llbracket 1,Q\rrbracket$. Since the Fourier coefficients of the signal are identical across channels, $\hat{\bd x}$ can be jointly estimated from the measurements by solving
\begin{equation}\label{eq:block_estimation}
	\begin{bmatrix}
	\bd f^{(1)} \\ \bd f^{(2)} \\ \vdots \\ \bd f^{(Q)}
	\end{bmatrix} =
	\begin{bmatrix}
	\bd G^{(1)} \\ \bd G^{(2)} \\ \vdots \\ \bd G^{(Q)}
	\end{bmatrix}
	\hat{\bd x}\triangleq \bd G\hat{\bd x}.
\end{equation}
The system will admit a unique solution $\hat{\bd x}$ when $\bd G$, which has size $\left(\sum_{i=1}^{Q}L^{(i)}\right)\times (2K+1)$, has full column-rank. This happens when the matrix is tall and when no two channels have identical trigger times, which can be achieved by setting different values for the temporal contrast threshold across channels. Effectively, the scheme allows for a reduction in the sampling requirement in each channel. To obtain the minimal number of measurements, we set the temporal contrast threshold to a value greater than $Q$ times the critical threshold stated in Proposition~\ref{prop:sparse_bound}. This result also readily extends to nonuniform $\mathrm L$-splines. We summarize the result in the following proposition.
\begin{proposition}[Perfect reconstruction of FRI signals from SIMO neuromorphic encoding]\label{prop:simo_sparse_bound}
	The FRI signal $x(t)$ in Eq.~\eqref{eq:pulse_signal_model} can be perfectly recovered from the events $\sT_{C^{(i)}}\{f\}$, where $f(t) = (x*g)(t)$, where $g(t)$ satisfies the alias-cancellation conditions (cf. Eq.~\eqref{eq:alias_cancellation}) and the temporal contrast thresholds satisfy
	\[
		0 < C^{(j)} \neq C^{(i)} < \frac{Q(f_{\max} - f_{\min})}{2K+1}, \; j\neq i,
	\]
	where $\displaystyle f_{\max} = \max_{t\in [0,T]} f(t)$ and $\displaystyle f_{\min} = \min_{t\in [0,T]} f(t)$.
\end{proposition}


\begin{figure*}[!t]
\setcounter{subfigure}{-2}
	\centering
	\subfigure{\label{fig:}\includegraphics[height=.215in]{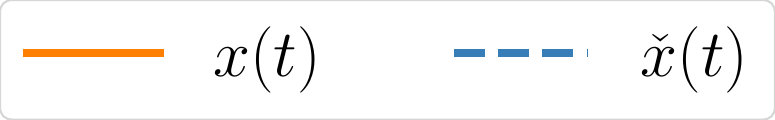}}
	\subfigure{\label{fig:}\includegraphics[height=.215in]{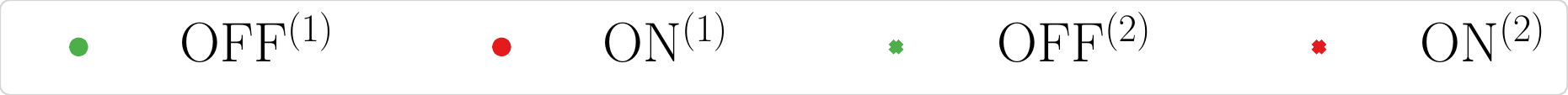}}
	\subfigure[Stream of Dirac impulses]{\label{fig:}\includegraphics[width=3.2in]{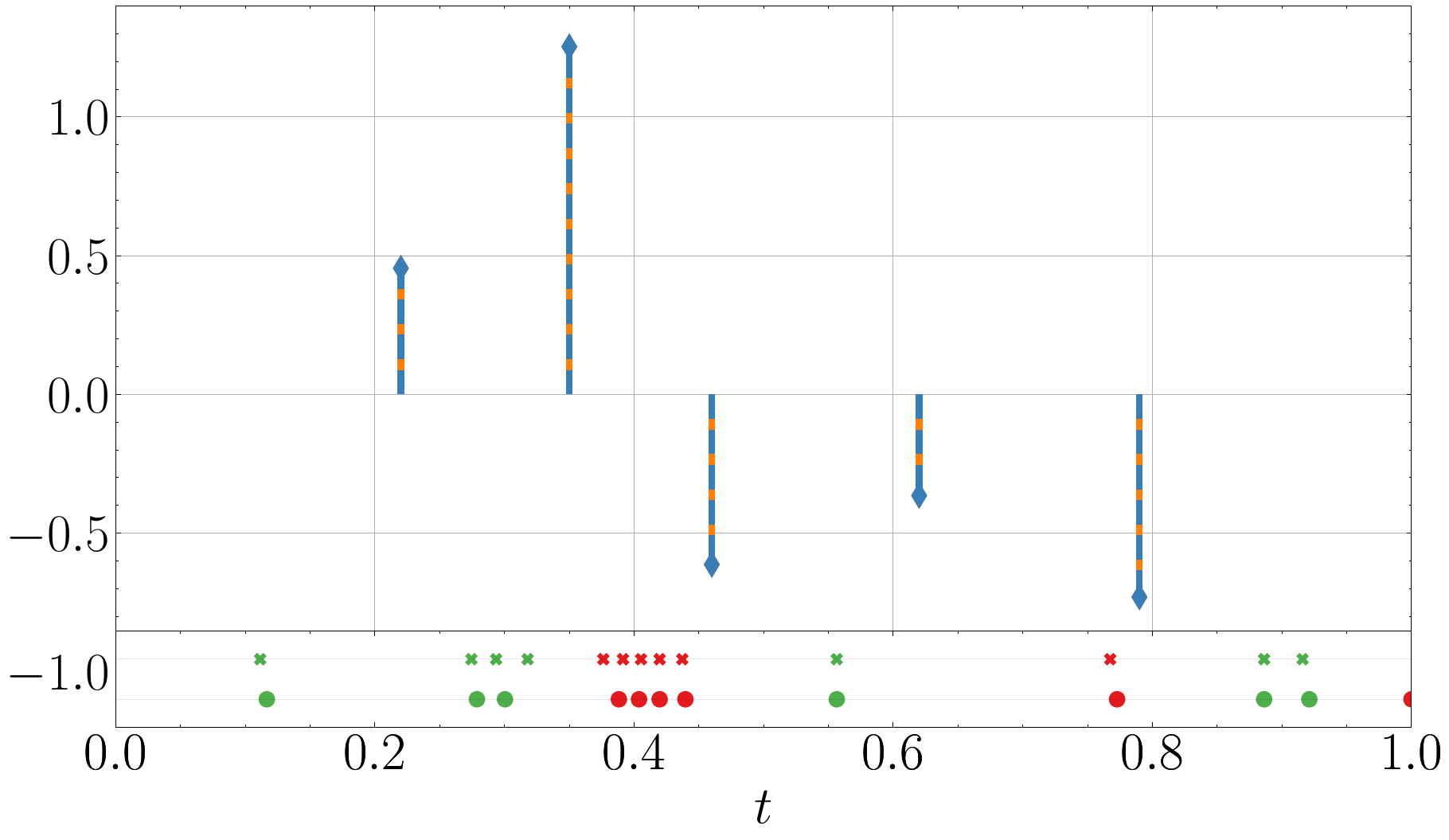}}
	\subfigure[Stream of cubic B-spline pulses]{\label{fig:}\includegraphics[width=3.2in]{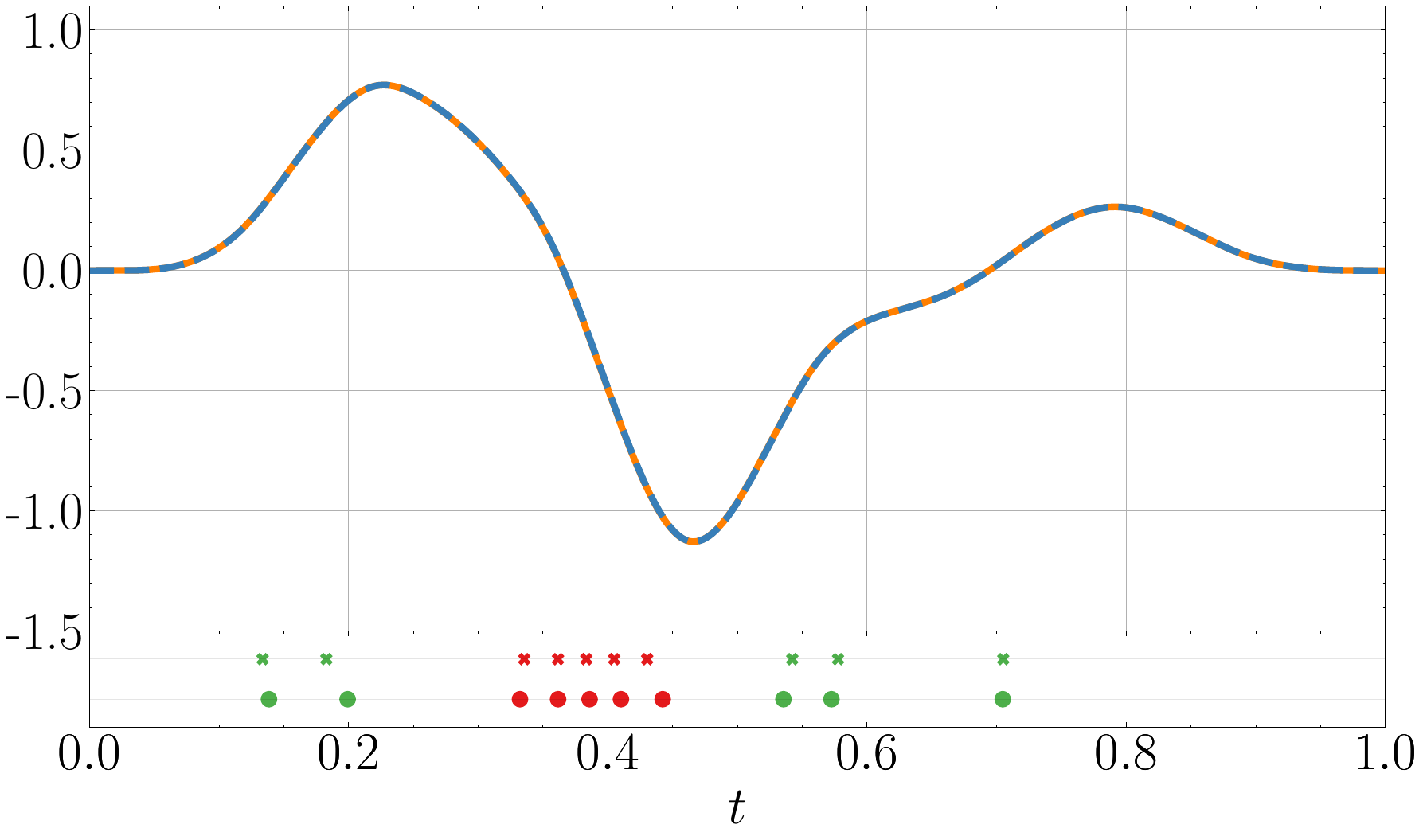}}
	\subfigure[A piecewise-constant signal]{\label{fig:}\includegraphics[width=3.2in]{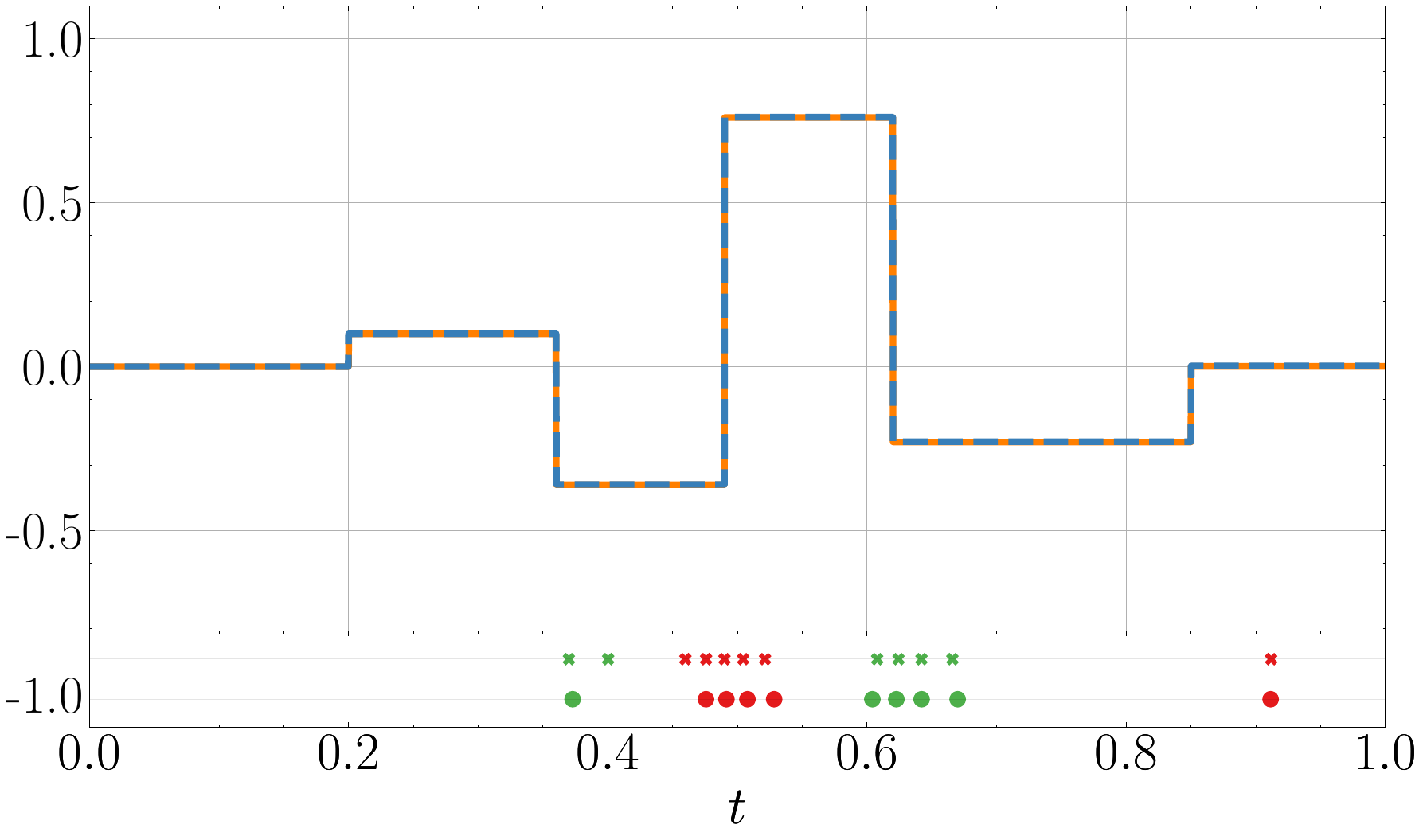}}
	\subfigure[A piecewise-linear signal]{\label{fig:}\includegraphics[width=3.2in]{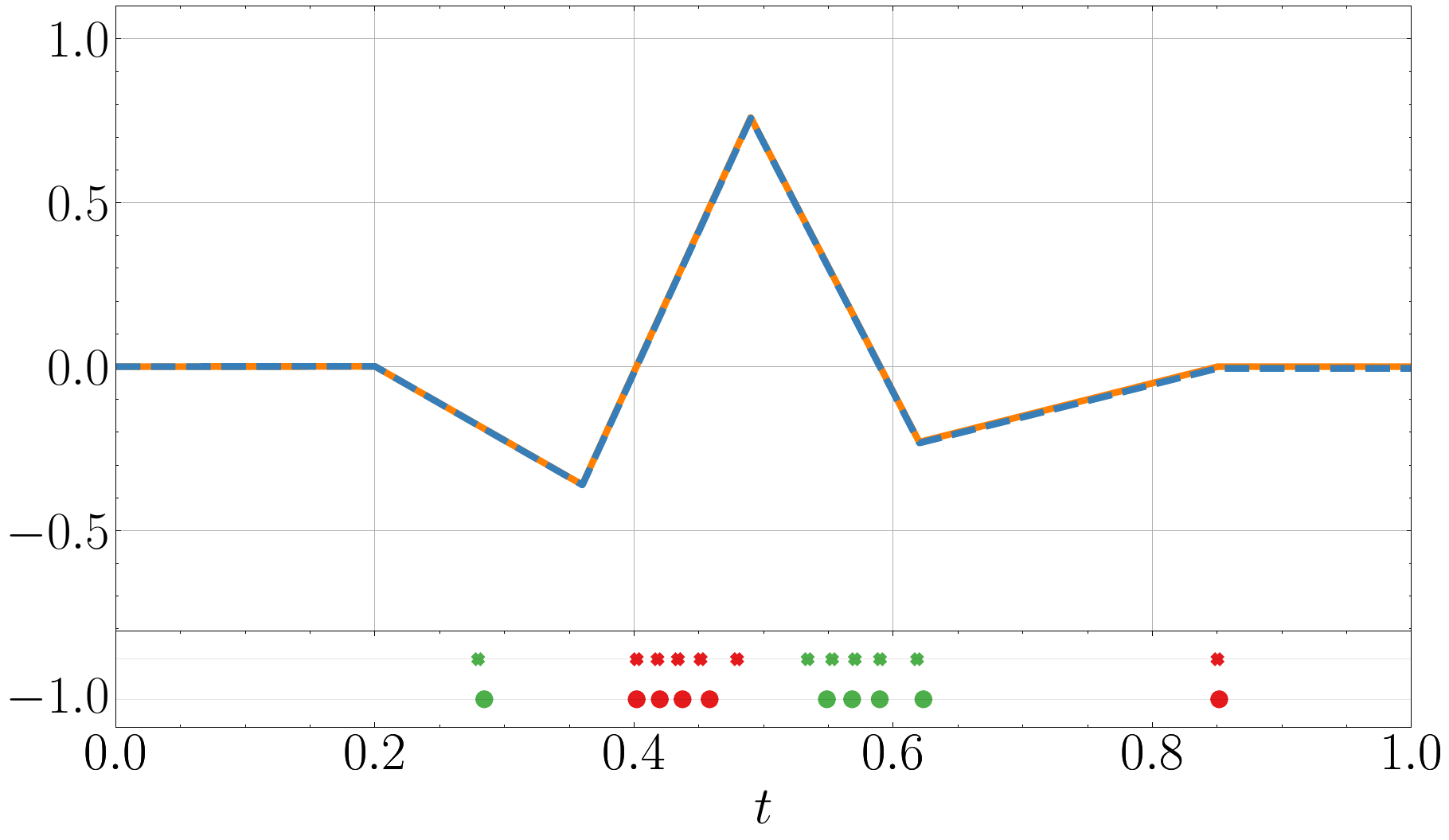}}
	\caption{SIMO neuromorphic sampling of FRI signals, using $Q=2$ channels, and the zeroth-order SMS sampling kernel $g(t)$ and perfect reconstruction using Prony's method for various examples. The FRI signal $x(t)$, the reconstruction $\check{x}(t)$; and the ON (green) and OFF (red) events $\{(t^{(i)}_1,p^{(i)}_1), (t^{(i)}_2,p^{(i)}_2), \ldots\}$ in channels $i=1,2$ are also shown.}
	\label{fig:SIMO_reconstruction}
\end{figure*}
\begin{figure*}[!t]
\setcounter{subfigure}{-2}
	\centering
	\subfigure{\label{fig:}\includegraphics[height=.2in]{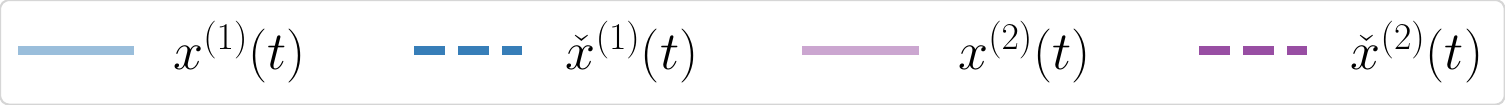}}
	\subfigure{\label{fig:}\includegraphics[height=.2in]{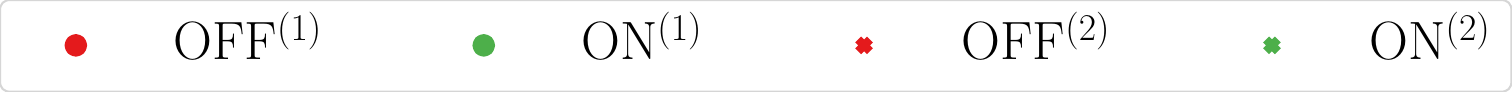}}
	\subfigure[Stream of Dirac impulses]{\label{fig:}\includegraphics[width=3.2in]{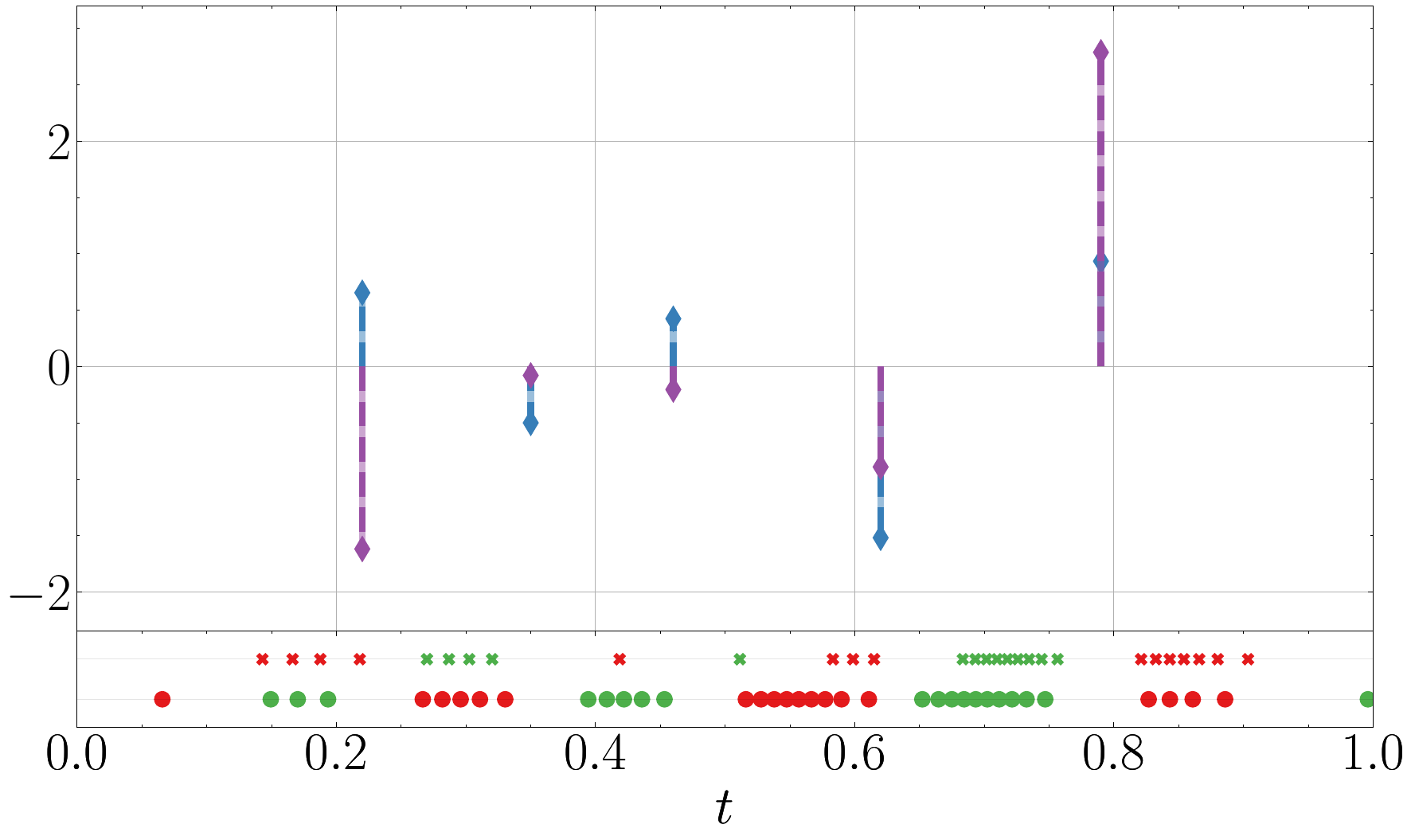}}
	\subfigure[Stream of cubic B-spline pulses]{\label{fig:}\includegraphics[width=3.2in]{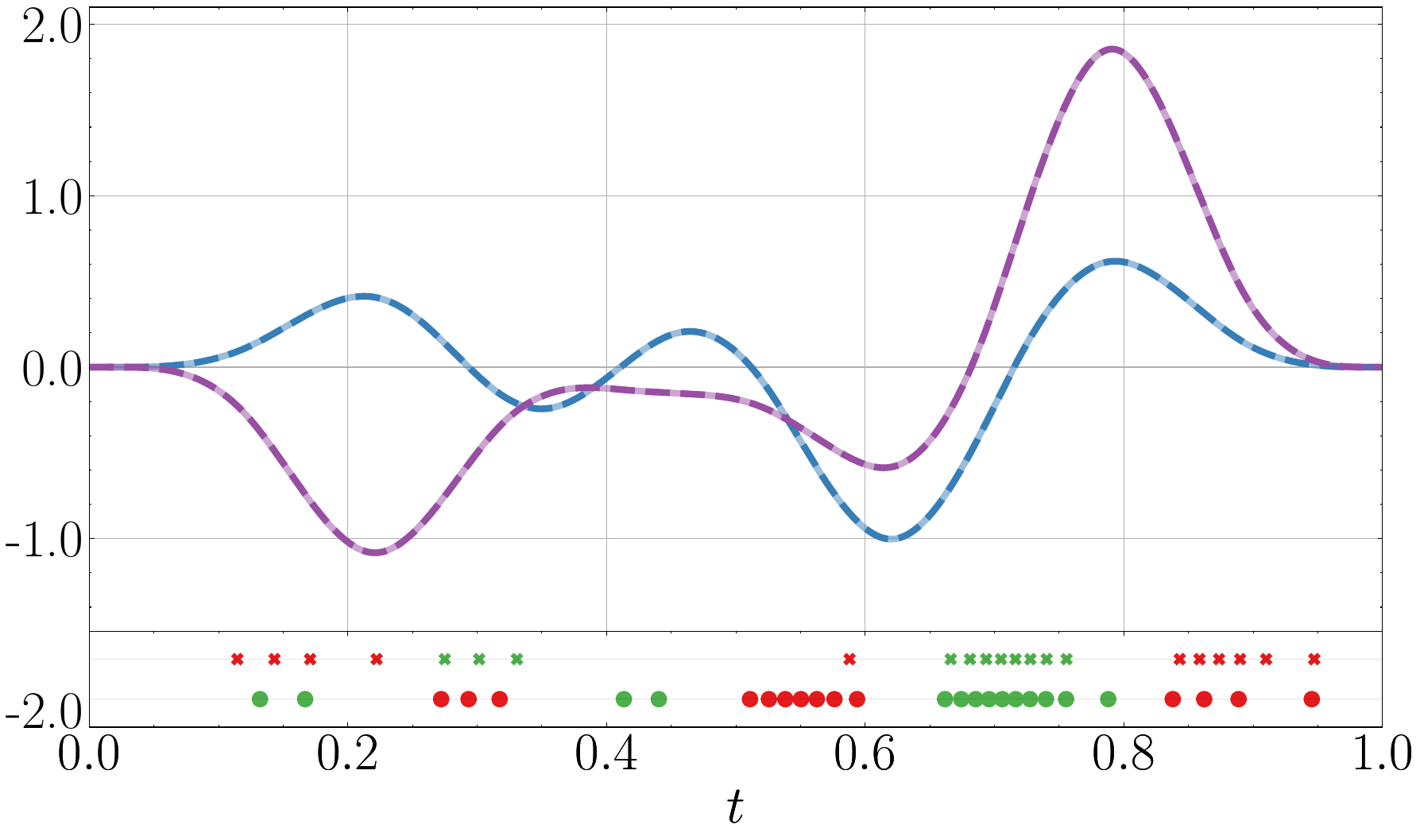}}
	\subfigure[A piecewise-constant signal]{\label{fig:}\includegraphics[width=3.2in]{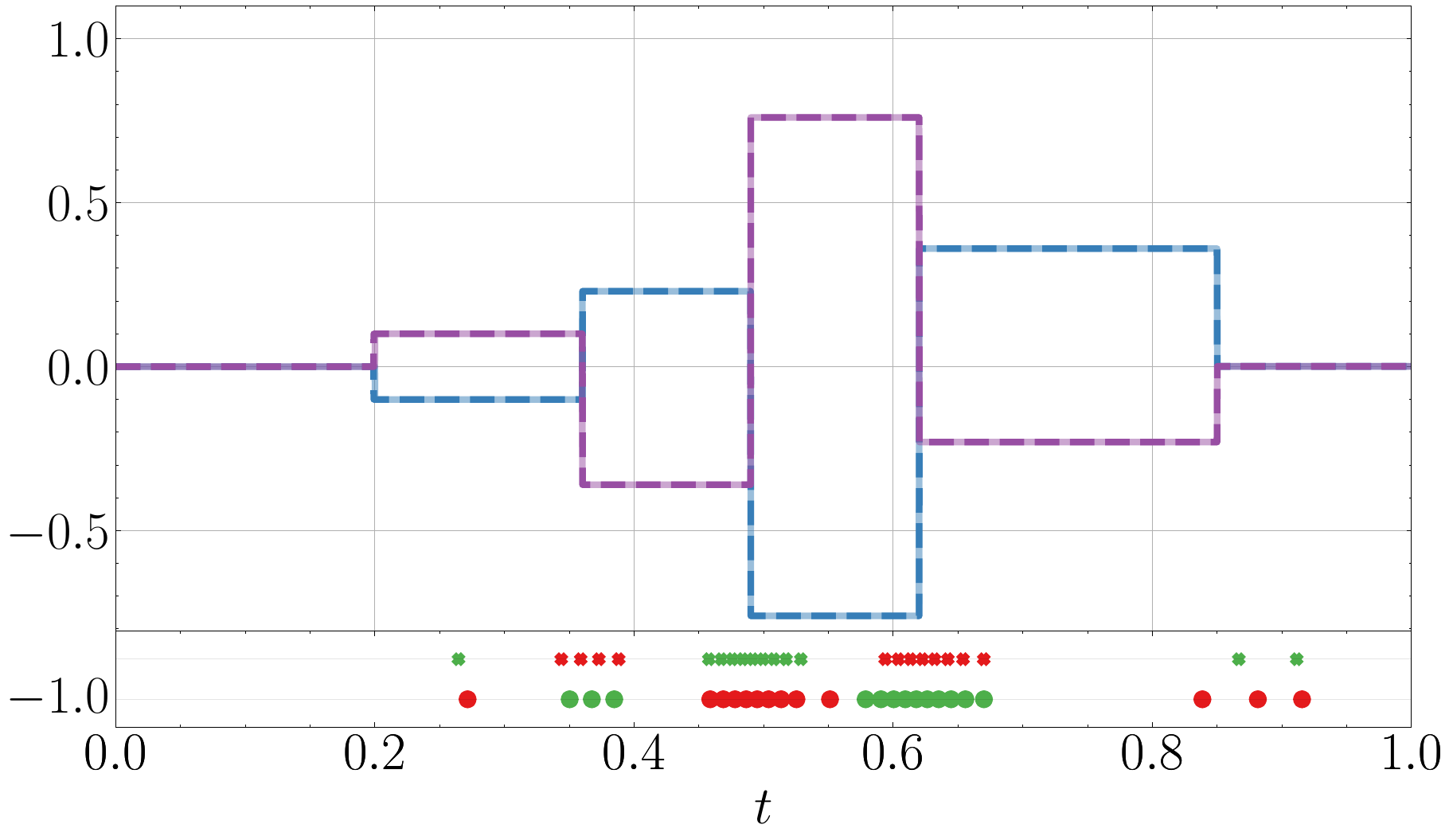}}
	\subfigure[A piecewise-linear signal]{\label{fig:}\includegraphics[width=3.2in]{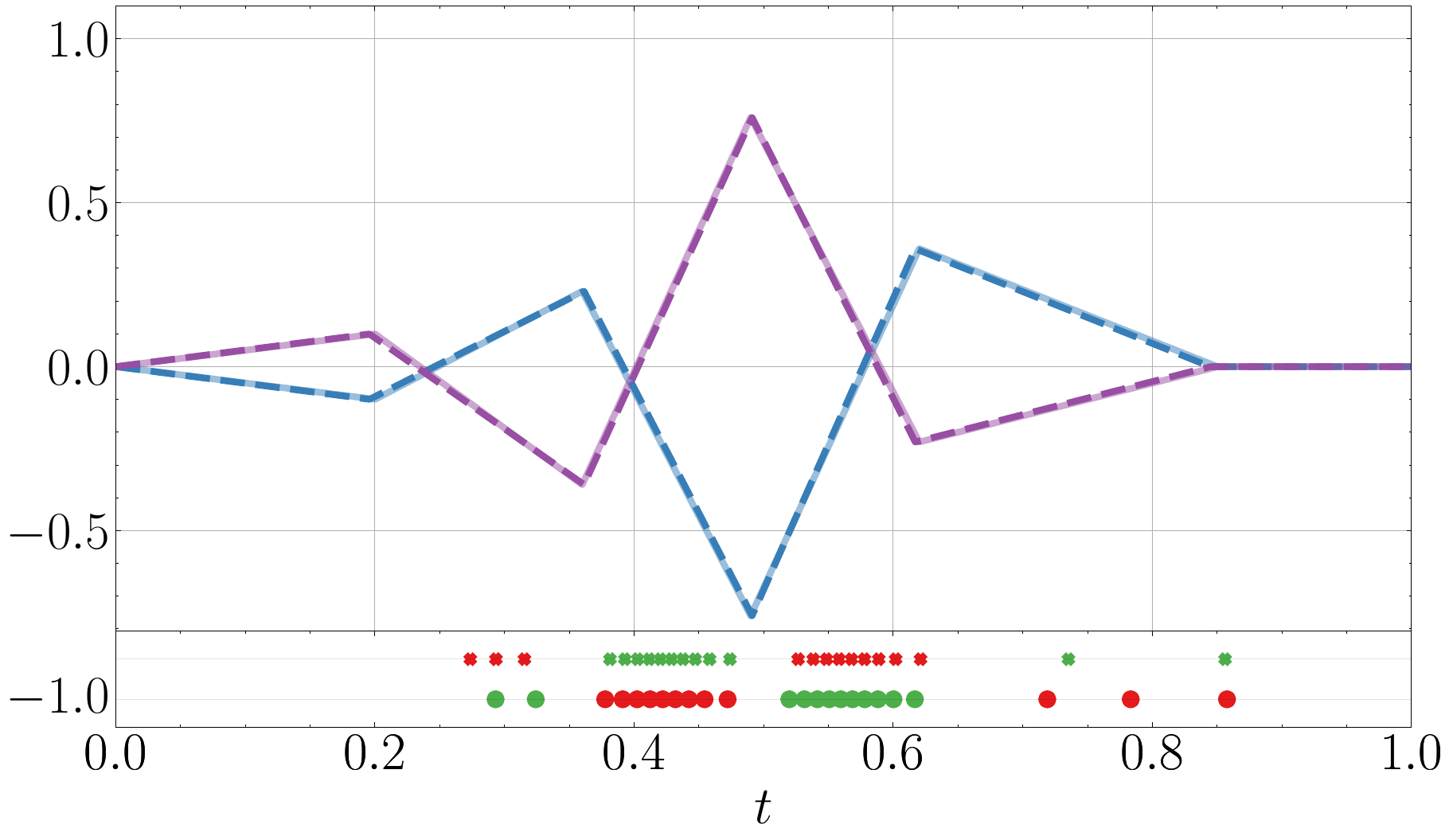}}
	\caption{MIMO neuromorphic sampling of FRI signals, using $Q=2$ channels, and the zeroth-order SMS sampling kernel $g(t)$ and perfect reconstruction using Prony's method for various examples. The FRI signals $\{x^{(i)}(t)\}$, the reconstructions $\{\check{x}^{(i)}(t)\}$; and the ON (green) and OFF (red) events $\{(t^{(i)}_1,p^{(i)}_1), (t^{(i)}_2,p^{(i)}_2), \ldots\}$ for channels $i=1,2$ are also shown.}
	\label{fig:MIMO_reconstruction}
\end{figure*}
\subsection{MIMO Neuromorphic Sampling}
\label{subsec:mimo_sampling}
Next, consider the MIMO configuration, where we have a vector input signal with entries $x^{(i)}\in L_2([0,T])$ given as
\begin{equation}\label{eq:mimo_signal_model}
	x^{(i)}(t) = \sum_{k=0}^{K-1} a^{(i)}_k\varphi(t-\tau_k),
\end{equation}
where $i\in\llbracket 1,Q\rrbracket$ denotes the channel index, and $Q$ is the total number of neuromorphic encoders. (cf. Figure~\ref{fig:mimo_encoding_schematic}). We allow the coefficients $\bd a^{(i)} = [a^{(i)}_0\;a^{(i)}_1\cdots a^{(i)}_{K-1}]^\TT \in \bb R^K$ of each entry to vary, whilst the pulse $\varphi$ and the support parameters $\bld \tau = [\tau_0\;\tau_1\;\cdots\;\tau_{K-1}]^\TT\in\bb R^K$ remain fixed across the channels, {\it i.e.}, the vector input signal has common support parameters. This is akin to the sparse common support FRI (SCS-FRI) signal model considered in \cite{hormati2011compressive}, which is encountered in practical applications such as pulsed Doppler radar \cite{rudresh2017radar}. Each of the entries in the vector input has a rate of innovation of $\displaystyle\frac{2K}{T}$, and can be perfectly recovered from its events using Algorithm~\ref{algo:sparse}. However, since the support parameters are identical, the annihilating filter is common across the channels and the common support parameters can be estimated using the \emph{block-annihilation} technique \cite{hormati2011compressive}.\\
\indent Let $f^{(i)}(t) = \left(x^{(i)}*g\right)(t)$ denote the filtered signal in the $i$\textsuperscript{th} channel, and $\sT_{C^{(i)}}\{f^{(i)}\} = \left\{\left(t^{(i)}_m,p^{(i)}_m\right)\right\}_{m=0}^{L^{(i)}}$ be the events generated by the neuromorphic encoder in the $i$\textsuperscript{th} channel with temporal contrast threshold $C_i$. Here, $L^{(i)}$ denotes the number of measurements obtained in the $i$\textsuperscript{th} channel. Using Lemma~\ref{lem:ttransform}, we obtain the amplitude samples of the input vector signal as $\bd f^{(i)} = \left[f^{(i)}\left(t^{(i)}_1\right)\; f^{(i)}\left(t^{(i)}_2\right) \cdots f^{(i)}\left(t^{(i)}_{L^{(i)}}\right)\right]^\TT$, which are linearly related to the Fourier coefficients of the FRI signal, {\it i.e.}, $\bd f^{(i)} = \bd G^{(i)}\hat{\bd x}^{(i)}, \; i\in\llbracket 1,Q\rrbracket$. Using Lemma~\ref{lem:ctemMatrix}, the linear systems are left-invertible when $L^{(i)}\geq 2K+1, \; i\in\llbracket 1,Q\rrbracket$. Once the Fourier coefficients are obtained, we employ the block-annihilation technique, {\it i.e.}, we have a $(K+1)$-tap filter $\bd h$ that satisfies
\[
	\left(\Gamma_K\hat{\bd x}^{(i)}\right)\bd h = \bd 0,\;\forall i\in\llbracket 1,Q\rrbracket.
\]
Consolidating the annihilation property across channels, and invoking the common support property results in the block-annihilation model:
\begin{equation}\label{eq:block_annihilation}
	\bld\Gamma\bd h = \begin{bmatrix}
	\Gamma_K\hat{\bd x}^{(1)} \\
	\Gamma_K\hat{\bd x}^{(2)} \\
	\vdots \\
	\Gamma_K\hat{\bd x}^{(Q)}
	\end{bmatrix}\bd h = \bd 0.
\end{equation}
It has been shown in \cite{hormati2011compressive} that block-annihilation provides a superior estimate of the support in the presence of measurement noise, compared with averaging the channel-wise estimates of the support. The coefficients $\bd a^{(i)}$ can be found using least-squares regression of $\hat{\bd x}^{(i)} = \bd S\bd V\bd a^{(i)}, \; i\in\llbracket 1,Q\rrbracket$.\\
\indent The block-annihilation technique can be readily extended to estimating the common knot parameters of nonuniform $\mathrm L$-splines as in Section~\ref{subsec:nonuniform_spline_reconstruction}, with an appropriate change in the definition of the matrix $\bd S$.\\
\indent In each channel, we require $L^{(i)}\geq 2K+1$ events for perfect reconstruction. We obtain the minimal number of events by setting the temporal contrast threshold of the encoder in each of the channels similar to Proposition~\ref{prop:sparse_bound}. The preceding discussion is encapsulated in the form of the following proposition.
\begin{proposition}[Perfect reconstruction of FRI signals from MIMO neuromorphic encoding]\label{prop:mimo_sparse_bound}
	Vector FRI signals with entries $x^{(i)}(t)$ as in Eq.~\eqref{eq:mimo_signal_model} can be perfectly recovered from the events $\sT_{C^{(i)}}\{f^{(i)}\}$, where $f^{(i)}(t) = (x^{(i)}*g)(t)$, where $g(t)$ satisfies the alias-cancellation conditions (cf. Eq.~\eqref{eq:alias_cancellation}) and the temporal contrast thresholds satisfy
	\[
		0 < C^{(i)} < \frac{f^{(i)}_{\max} - f^{(i)}_{\min}}{2K+1}, \; \forall i\in\llbracket 1,Q\rrbracket,
	\]
	where $\displaystyle f^{(i)}_{\max} = \max_{t\in [0,T]} f^{(i)}(t)$ and $\displaystyle f^{(i)}_{\min} = \min_{t\in [0,T]} f^{(i)}(t)$.
\end{proposition}


\subsection{Experimental Results}
Consider SIMO encoding of the FRI signal $x(t)$ in Eq.~\eqref{eq:pulse_signal_model} with $Q=2$ channels in the SIMO configuration. Consider two choices of $\varphi$: a Dirac impulse, and $\varphi(t)=\beta^{(3)}(\frac{1}{10})$, which is a time-scaled cubic B-spline; and nonuniform $\mathrm D^j$-splines in Eq.~\eqref{eq:nonuniform_Lsplines}, for $j=1,2$. The shift parameters are drawn from a uniform distribution over the interval $[0,1]$, and the coefficients are drawn from the standard normal distribution $\cl N(0,1)$. We use the zeroth-order SMS kernel for sampling and encoding using two neuromorphic encoders with the parameters critically set according to Proposition~\ref{prop:simo_sparse_bound}. Figure~\ref{fig:SIMO_reconstruction} shows the input signal $x(t)$, and the reconstruction $\check{x}(t)$, along with the event time-instants and event polarities obtained across two channels $\{(t^{(i)}_1,p^{(i)}_1), (t^{(i)}_2,p^{(i)}_2),\ldots\}$, $i=1,2$. The recovered parameters are accurate up to numerical precision, indicating perfect reconstruction.\\
\indent Next, consider MIMO encoding of $Q=2$ FRI signals $x^{(i)}(t),\;i=1,2$, with the same rate of innovation, common support $\bld \tau$ drawn uniformly at random over the interval $[0,1]$, and the coefficients drawn from $\cl N(0,1)$ different across channels. The choices of $\varphi$ are the same as in the SIMO setting. We use the zeroth-order SMS kernel for sampling, using two neuromorphic encoders with parameters set critically according to Proposition~\ref{prop:mimo_sparse_bound}. Figure~\ref{fig:MIMO_reconstruction} shows the input signals $x^{(i)}(t)$, and the reconstructions $\check{x}^{(i)}(t)$, along with the event time-instants and event polarities obtained across two channels $\{(t^{(i)}_1,p^{(i)}_1), (t^{(i)}_2,p^{(i)}_2),\ldots\}$, $i=1,2$. The recovered parameters for each signal are accurate up to numerical precision, indicating perfect reconstruction.



 
\section{Conclusions}
We introduced the novel paradigm of neuromorphic sampling of FRI signals, thereby connecting sparse signals to sparse sampling. The sampling is inherently opportunistic, {\it i.e.}, the events are recorded only when there is a significant change in the signal. Using Fourier-domain analysis, we showed that perfect signal reconstruction is possible using parameter estimation, when the measurements are of the order of the rate of innovation of the signal. The estimation of the shift parameters is performed using Prony's method, and the estimation of the coefficients is performed using linear least-squares regression. We provided sufficient conditions on the temporal contrast threshold of the neuromorphic encoder to ensure perfect signal reconstruction. Further, the framework is readily extendable to multichannel sampling in the SIMO and MIMO configurations. In the MIMO configuration, the support parameters can be jointly estimated using the block-annihilation technique, and, in the SIMO configuration, both the coefficients and shift parameters can be jointly estimated to reduce the sampling requirement in each channel as compared to the single channel case. We verified the claims and demonstrated perfect signal reconstruction using numerical experiments.\\
\indent Immediate directions in which the work proposed in this paper can be extended are developing applications to imaging modalities such as radar, sonar, and ultrasound. Furthermore, in practical applications, one must consider the effect of noise and model mismatch, which are fertile directions for further research in this area.


\appendices


\bibliographystyle{ieeetr}
\bibliography{references.bib}

\end{document}